\documentclass[runningheads]{llncs}
\usepackage{graphicx} % Required for inserting images
\usepackage{xcolor}
\usepackage{microtype}
\usepackage{amsfonts}
\usepackage{cite}
\usepackage{amsmath,amssymb,amsfonts}
\usepackage{mathtools}
\usepackage{mdframed}
\usepackage{enumitem}
\usepackage[colorlinks,linkcolor={blue!50!black},
    citecolor={blue!50!black},
    urlcolor={blue!80!black}]{hyperref}
\usepackage{tikz}
\usepackage{listings}
\usepackage{caption}
\usepackage{wrapfig}
\usepackage{subcaption}
\usepackage{multirow}
\usepackage[algoruled, linesnumbered, noend]{algorithm2e}
\SetKwIF{IfNot}{ElseIfNot}{}{if not}{then}{else if not}{}{}

\SetCommentSty{mycommfont}
\usepackage[misc]{ifsym}

\usepackage{booktabs,siunitx}
\usepackage{pgf}

\usetikzlibrary{backgrounds,patterns,patterns.meta,arrows,arrows.meta,calc,shapes,shadows,decorations.pathmorphing,decorations.pathreplacing,automata,shapes.multipart,positioning,shapes.geometric,fit,circuits,trees,shapes.gates.logic.US,fit, shadows.blur,shapes.symbols,intersections}

\usepackage{etoolbox} % newbool, setbool, ifbool

\newbool{artifactfunctional}
\newbool{artifactreusable}
\newbool{artifactavailable}
\newcommand{\ensuremathspace}[1]{\ensuremath{#1}\xspace}
\newcommand{\mathsymbol}[2]{ \newcommand{#1}{\ensuremathspace{#2}} }

\mathsymbol{\naturals}{\mathbb{N}}
\mathsymbol{\integers}{\mathbb{Z}}
\mathsymbol{\reals}{\mathbb{R}}
\mathsymbol{\unitinterval}{[0,1]}
\mathsymbol{\supp}{\mathrm{supp}}
\mathsymbol{\Distr}{Distr}

\mathsymbol{\sinit}{s_{0}}
\mathsymbol{\Act}{Act}
\mathsymbol{\act}{\alpha}
\mathsymbol{\mpm}{P}
\mathsymbol{\mdp}{M}
\mathsymbol{\mdpT}{(S,\sinit,Act,\mpm)}
\mathsymbol{\mcT}{(S,\sinit,\mpm)}

\mathsymbol{\actrandom}{\act_{\mathrm{rand}}}
\newcommand{\sched}[1][]{\ensuremathspace{ \sigma_{#1} }}
\mathsymbol{\schedopt}{\sched[]^{*}}
\newcommand{\schedulers}[1][M]{\Sigma^{#1}}
\mathsymbol{\imc}{\mdp^{\sched}}

\mathsymbol{\variables}{\mathcal{V}}
\mathsymbol{\tree}{\mathcal{T}}
\mathsymbol{\treeT}{(T, \gamma, \delta)}
\mathsymbol{\predicates}{\Psi_{\variables}}

\mathsymbol{\paths}{Path(T)}
\newcommand{\leaf}[1][\tree]{\ensuremathspace{ \mathsf{leaf}^{#1}} }

\mathsymbol{\template}{\mathbb{T}}
\mathsymbol{\templateT}{(T, \Gamma, \Delta)}

\mathsymbol{\decisionfun}{\mathcal{D}}
\mathsymbol{\boundfun}{\mathcal{B}}
\mathsymbol{\actionfun}{\mathcal{A}}
\newcommand{\params}[1][\template]{\mathcal{X}^{#1}}

\newcommand{\fpol}[1][]{\mathsf{pol}_{#1}}
\newcommand{\fdom}{\mathsf{dom}}
\mathsymbol{\fact}{\mathsf{act}}
\mathsymbol{\fsel}{\mathsf{sel}}
\mathsymbol{\fharm}{\mathsf{harm}}

\newcommand{\sts}[2]{\mathsf{st}^{#1}({#2})}
\mathsymbol{\meetsin}{\bowtie_{\;\!i}^{\:\!n}}

\mathsymbol{\unsatcore}{\mathsf{UC}}
\mathsymbol{\paramsuc}{\params[\unsatcore]}
\mathsymbol{\statepathuc}{\mathit{SP}^{\unsatcore}}
\mathsymbol{\criticalstates}{S^{\unsatcore}}
\mathsymbol{\criticalsched}{\sched[\downarrow \criticalstates]}

\mathsymbol{\pf}{\varphi}

\mathsymbol{\family}{\mathcal{F}}
\mathsymbol{\ftemplate}{\template(\family)}
\mathsymbol{\superfamily}{ \family^{\template} }

\mathsymbol{\kmone}{{k{-}1}}

\newcommand{\F}[1]{\Diamond #1}

\newcommand{\prob}[2][]{\mathbb{P}\left[ #1 \models #2\right ]}

\newcommand{\probmax}[2][]{\mathbb{P}_{\max}\left[ #1 \models #2\right ]}

\newcommand{\tool}[1]{{\textsc{#1}}}
\newcommand{\storm}{\tool{Storm}\xspace}
\newcommand{\paynt}{\tool{Paynt}\xspace}
\newcommand{\dtcontrol}{\tool{dtControl}\xspace}
\newcommand{\omdt}{\tool{OMDT}\xspace}

\newcommand{\dtpaynt}{\tool{dtPaynt}\xspace}

\def\orcidID#1{\smash{\href{http://orcid.org/#1}{\protect\raisebox{-1.25pt}{\protect\includegraphics{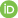}}}}}

\renewcommand{\paragraph}[1]{\smallskip\noindent\emph{#1}}
\renewcommand{\subsubsection}[1]{\medskip\noindent\textbf{#1}}

\newlength\myheight
\newlength\mydepth
\settototalheight\myheight{Xygp}
\newcommand*\inlinegraphics[1]{%
  \settototalheight\myheight{Xygp}%
  \settodepth\mydepth{Xygp}%
  \raisebox{-\mydepth}{\includegraphics[height=\myheight]{#1}}%
}

\definecolor{prismgreen}{HTML}{009900}
\definecolor{prismred}{HTML}{cc0000}
\definecolor{prismblue}{HTML}{0000cc}

\lstdefinelanguage{Prism}{
        basicstyle=\color{prismred}\scriptsize\ttfamily,
        literate=*	{0}{{\textcolor{prismblue}{0}}}{1}
			{1}{{\textcolor{prismblue}{1}}}{1}
			{2}{{\textcolor{prismblue}{2}}}{1}
			{3}{{\textcolor{prismblue}{3}}}{1}
			{4}{{\textcolor{prismblue}{4}}}{1}
			{5}{{\textcolor{prismblue}{5}}}{1}
			{6}{{\textcolor{prismblue}{6}}}{1}
			{7}{{\textcolor{prismblue}{7}}}{1}
			{8}{{\textcolor{prismblue}{8}}}{1}
			{9}{{\textcolor{prismblue}{9}}}{1}
			{.0}{{\textcolor{prismblue}{.0}}}{2}
			{.1}{{\textcolor{prismblue}{.1}}}{2}
			{.2}{{\textcolor{prismblue}{.2}}}{2}
			{.3}{{\textcolor{prismblue}{.3}}}{2}
			{.4}{{\textcolor{prismblue}{.4}}}{2}
			{.5}{{\textcolor{prismblue}{.5}}}{2}
			{.6}{{\textcolor{prismblue}{.6}}}{2}
			{.7}{{\textcolor{prismblue}{.7}}}{2}
			{.8}{{\textcolor{prismblue}{.8}}}{2}
			{.9}{{\textcolor{prismblue}{.9}}}{2}
			{[}{{\textcolor{black}{[}}}{1}
			{]}{{\textcolor{black}{]}}}{1}
			{+}{{\textcolor{black}{+}}}{1}
			{-}{{\textcolor{black}{-}}}{1}
			{=}{{\textcolor{black}{=}}}{1}
			{<}{{\textcolor{black}{<}}}{1}
			{>}{{\textcolor{black}{>}}}{1}
			{\&}{{\textcolor{black}{\&}}}{1}
			{|}{{\textcolor{black}{|}}}{1}
			{:}{{\textcolor{black}{:}}}{1}
			{;}{{\textcolor{black}{;}}}{1}
			{(}{{\textcolor{black}{(}}}{1}
			{)}{{\textcolor{black}{)}}}{1}
			{..}{{\textcolor{black}{..}}}{2},
        keywords= {bool,ceil,const,ctmc,double,dtmc,endinit,endmodule,endrewards, endsystem,F,false,floor,formula,G,global,I,init,int,label,max,mdp,min,module,nondeterministic,P,Pmin,Pmax,prob,probabilistic,rate,rewards,Rmin,Rmax,S,stochastic,system,true,U, option, either, assignment, relation, operation, hole, variable},
        keywordstyle={\bfseries\color{black}},
        numberstyle=\footnotesize\color{black},
        comment=[l] {//}, morecomment=[s]{/*}{*/},
        commentstyle= \color{prismgreen},
        tabsize=4,
        captionpos=b,
        escapechar=^,
        moredelim=[is][\color{orange}]{@}{@},
}

\title{Small Decision Trees for MDPs\\with Deductive Synthesis~}

\author{Roman Andriushchenko\inst{1}\orcidID{0000-0002-1286-934X} \and  Milan \v{C}e\v{s}ka {(\Letter)} \inst{1}\orcidID{0000-0002-0300-9727} \and \\Sebastian Junges\inst{2}\orcidID{0000-0003-0978-8466}  \and Filip~Mac\'{a}k\inst{1}\orcidID{0009-0004-4277-2751}}

\authorrunning{Andriushchenko et al.}

\institute{Brno University of Technology, Brno, Czech Republic\\ \email{ceskam@fit.vut.cz} 
\and Radboud University, Nijmegen, the Netherlands}

\begin{document}

\maketitle

\begin{abstract}
Markov decision processes (MDPs) describe decision making subject to probabilistic uncertainty. 
A classical problem on MDPs is to compute a policy, selecting actions in every state, that maximizes the probability of reaching a dedicated set of target states.
Computing such policies in tabular form is efficiently possible via standard algorithms. However, for further processing by either humans or machines, policies should be represented concisely, e.g., as a decision tree.
This paper considers finding (almost) optimal decision trees of minimal depth and contributes a deductive synthesis approach. Technically, we combine pruning the space of concise policies with an abstraction-refinement loop with an SMT-encoding that maps candidate policies into decision trees. Our experiments show that this approach beats the state-of-the-art solver using an MILP encoding by orders of magnitudes. The approach also pairs well with heuristic approaches that map a fixed policy into a decision tree: for an MDP with 1.5M states, our approach reduces the size of the given tree by 90\%, while sacrificing only 1\% of the optimal performance.

\end{abstract}

\section{Introduction}
Markov decision processes (MDPs) are the ubiquitous model to describe sequential decision making under uncertainty: the outcomes of nondeterministic actions are determined by a probability distribution over the successor states. MDP policies resolve the nondeterminism and describe for each state which action to take.
A classical synthesis task in MDPs is to compute a policy that maximizes a given objective, such as: \emph{Given a set of goal states, find a maximizing policy}, i.e.~a policy ensuring that the goal is reached with the maximal probability. These policies are efficiently computed by probabilistic model checkers such as~\storm~\cite{STORM} or Prism~\cite{KNP11} or can be approximated using (deep) reinforcement learning techniques~\cite{SuttonB:rlbook,DBLP:conf/tacas/Junges0DTK16}.
These techniques apply to huge MDPs that are concisely represented, e.g. in the PRISM language.
The result is a policy represented either in tabular form, mapping states to an action, or as a neural network. While the tabular form is often prohibitively large for further analysis by machines or a human, neural networks are hard to analyze despite the tremendous progress in neural network verification. This observation has motivated the search for concise representations of policies, in particular in the form of programs~\cite{DBLP:conf/icml/VermaMSKC18,batz2024programmatic} or decision trees (DTs)~\cite{bastani2018verifiable,topin2021iterative,vos2023optimal}. \emph{The main contribution of this paper is an approach to synthesize policies that are optimal within a class of small DTs.}

\begin{figure}[t]
    \begin{subfigure}{.25\textwidth}
    \centering
    \includegraphics[width=0.95\textwidth]{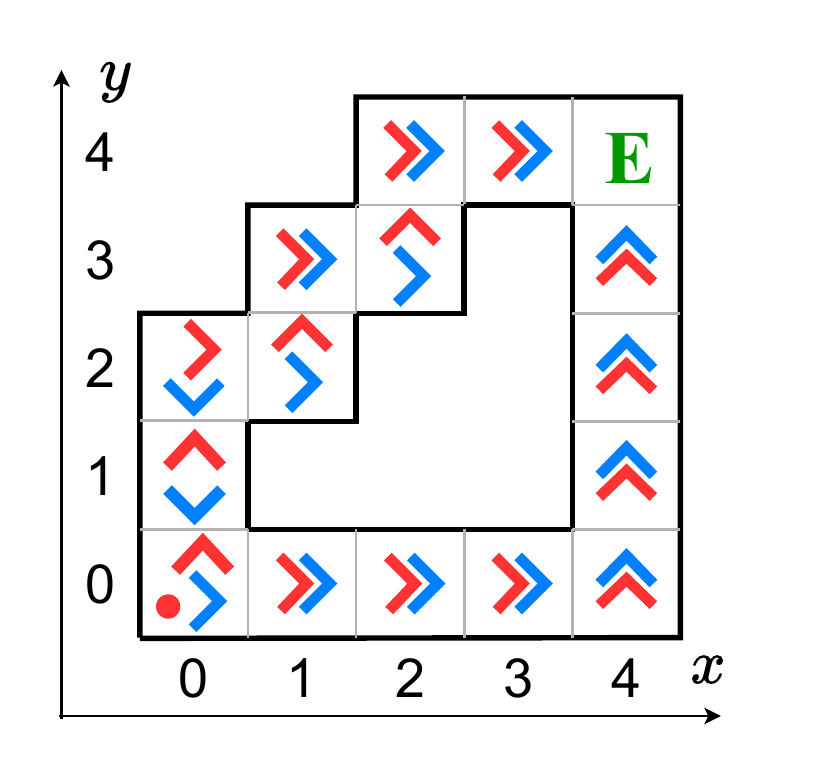}
    \caption{}
    \label{fig:example:maze-2}
    \end{subfigure}%
    \begin{subfigure}[b]{.55\textwidth}
    \centering
    \includegraphics[width=0.95\textwidth]{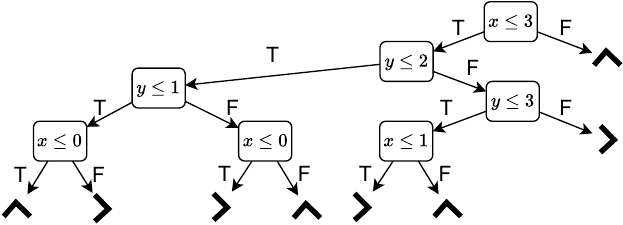}
    \caption{}
    \label{fig:example:maze:tree-4}
    \end{subfigure}%
    \begin{subfigure}[b]{.2\textwidth}
    \centering
    \includegraphics[width=0.95\textwidth]{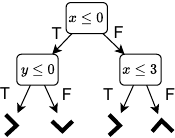}
    \caption{}
    \label{fig:example:maze:tree-2}
    \end{subfigure}%
    \vspace{-0.5em}
\caption{
(a)~A simple slippery maze. The goal is to lead the robot placed in the lower left corner (the red dot) towards the exit cell while minimizing the number of steps.
Red arrowheads illustrate the optimal policy that achieves a value of 12.8 (expected steps). Blue arrowheads illustrate a sub-optimal policy that achieves a value of 13.3.
(b)~The smallest DT implementing the optimal policy has depth 4.
(c)~The smallest DT implementing the sub-optimal policy has only depth 2.
}
\label{fig:example:maze}
\end{figure}

\paragraph{Illustrative example.} Consider a simple grid-world maze as in Fig.~\ref{fig:example:maze}. The agent starts at the bottom left and wants to reach the exit marked \textcolor{green!60!black}{\textbf{E}}. It can move in the four cardinal directions, and each action has a $10\%$ probability of transitioning into an unintended neighboring cell. Consequently, every state is reachable under every policy. When moving in a direction blocked by a wall, the agent bumps into the wall and remains in the same cell. E.g., in the cell $(x=1, y=2)$ moving to the right (as the blue policy does) means the agent will with high probability stay in the same cell but there is a small probability that it slips into the cell above. \storm computes an optimal (in this example, unique) policy that ensures reaching the exit in an expected 12.8 steps (the policy visualized by the red arrowheads).
To represent this policy as a DT  (using predicates that compare variables to constants) requires a tree of depth at least $4$ (see Fig.~\ref{fig:example:maze:tree-4}).
Alternatively, we may ask: What is the optimal expected number of steps to reach the exit among all policies that can be represented as a DT of depth $2$? The answer is $13.3$ realized by the policy visualized by the blue arrowheads. The corresponding DT is shown in Fig.~\ref{fig:example:maze:tree-2}. This policy aims to avoid the `staircase' in the left upper corner and then takes sub-optimal actions within that staircase.

\paragraph{Problem setup.}
We call a policy \emph{$k$-implementable}, if there is a DT of depth~$k$ (and with a particular class of predicates) that represents the policy. The first problem studied in the literature, \emph{the mapping problem}, asks whether a given tabular policy is $k$-implementable for any fixed $k$. Solving this problem then allows us to find the smallest DT, measured by depth, that implements this policy. 
The second problem studied, the \emph{synthesis problem}, is to find a policy that is optimal with respect to some objective, such as reaching the goal state, and within the class of $k$-implementable policies. We want to highlight that the mapping problem assumes one fixed policy. Therefore, the mapping problem cannot find the minimal representation of \emph{any} optimal policy. In particular, in our experiments, we show that we can find optimal policies that are $2$-implementable, whereas the policy that \storm computes is not $5$-implementable. The construction of optimal DTs is well-known to be NP-hard for different notions of optimality~\cite{hancock1996lower,laurent1976constructing}.

\paragraph{State-of-the-art: Policy mapping.}
Mapping policies into small, but not necessarily the smallest, DTs is prominently supported by the tools \mbox{\dtcontrol~\cite{ashok2021dtcontrol}} and Uppaal Stratego~\cite{david2015uppaal,ashok2019sos}. 
These tools approach the problem by \emph{learning} small DTs by recursively splitting the tree based on ideas like information gain. 
The result is an approximation or exact representation of the original policy. Generally, these tools favor scalability over minimality.

\paragraph{State-of-the-art: The synthesis problem.}
The tool
\omdt~\cite{vos2023optimal} builds a monolithic MILP that encodes both the structural constraints on the policy (being $k$-implementable) and the constraint that the policy achieves the optimal value (using the standard LP formulation for maximal discounted rewards). This LP-based approach encounters the same scalability as (MI)LP-based MDP model checking approaches face, e.g., in~\cite{DBLP:conf/atva/DehnertJWAK14,delgrange2020simple,drager2015permissive,andriushchenko2022inductive}.

\paragraph{Our approach: Abstraction refinement with SMT-based mapping.} Inspired by encodings of DTs in propositional formulas~\cite{narodytska2018learning},
we encode the set of $k$-implementable trees in an SMT formula over the bounded integers and with linear inequalities. Using an SMT solver allows us to determine if there is a $k$-implementable tree representing the given policy. Moreover, the encoding allows us to design an abstraction-refinement loop that avoids solving the synthesis problem in one shot. Our approach takes a set of $k$-implementable policies and abstracts them to search for an optimal policy in a larger class of policies.
If this policy is not improving over the best $k$-implementable policy found so far, it abandons the search here. Otherwise, by solving the mapping problem, it tests whether this optimal policy is $k$-implementable. If yes, we can abandon the search here and store the policy as our best policy so far. Otherwise, the policy is \emph{spurious}, and the search is recursively invoked on smaller subsets of $k$-implementable policies. 
\begin{example}
We present a conceptual version of our routine on the example given above, see Fig.~\ref{fig:example:maze}. To find an optimal $2$-implementable policy, we would first search for an optimal memoryless policy \sched. This policy is better than the previously found policy (we can~e.g.~initialize this policy as a random tree).
Using the mapping problem,  \sched is spurious.  We can now split and independently search for the best $k$-implementable policy that goes up in the initial state and for the best $k$-implementable policy that goes right in the initial state. We observe that these sub-classes can be overapproximated by memoryless deterministic policies on two sub-MDPs of the original MDP. 
\end{example}
\paragraph{Effective abstraction refinement.}
In Sec.~\ref{sec:refinement}, we introduce the abstraction-refinement loop \dtpaynt that analyzes the spurious policies to split in an informed way. 
Compared to an abstraction-refinement loop for finite state controllers in POMDPs~\cite{andriushchenko2022inductive}, the set of $k$-implementable policies is highly irregular and a policy can be represented by many different trees. To overcome this problem, we use an unsatisfiable core $\unsatcore$ witnessing that the given policy is not $k$-implementable. We introduce a \emph{harmonization} technique that analyzes $\unsatcore$ and finds two trees that serve as good approximations of the policy and, most importantly, they differ in one parameter that provides a good heuristic on how to construct the subsets of $k$-implementable policies. 
The proposed refinement procedure furthermore entails the ability to first learn a $k$-tree before learning a $k{+}1$-tree.

\paragraph{Relation to partially observable MDPs.}
Finding optimal $k$-implementable policies can be phrased as finding a colouring of states and for every colour an action. This reformulation clarifies a connection to the synthesis for memoryless observation-based policies POMDPs~\cite{li2011finding,kumar2015history,andriushchenko2022inductive}. Usually, in POMDPs, one cannot pick the colouring of the states, but such variations have been investigated in \cite{konsta2024should,DBLP:conf/aips/ChatterjeeCT18}. However, contrary to those settings, in DTs, the state colouring cannot be arbitrary but must be implementable with a DT. 

\paragraph{Contributions.} 
We propose \dtpaynt, an abstraction-refinement loop that iteratively invokes an SMT-based routine to search for DTs with maximum value among all DTs up to the given depth. 
\dtpaynt significantly outperforms  \omdt~\cite{vos2023optimal}, the state-of-the-art tool solving the same problem using a MILP encoding. In contrast to \dtcontrol, a prominent tool for policy mapping, \dtpaynt is able to effectively control the trade-offs between the size and quality of the resulting DTs. For example, for the \emph{consensus} protocol, \dtpaynt finds a DT with 10 inner nodes, which is about 
38 times smaller than the DT found by \dtcontrol, while it achieves  93\% of the optimal performance. 
\dtpaynt can scale to MDPs with hundreds of thousands of states, provided that a small DT with the desired performance exists.
Finally, \dtpaynt can be used to reduce large DTs: for a variant of a \emph{cmsa} model with 1.5M states, we reduce the number of inner nodes in the DT constructed by \dtcontrol from 236 to 22 while sacrificing less than 1\% of the performance.

\section{Background and Problem Statement}

A \emph{distribution} over a countable set $A$ is a~function $\mu \colon A \rightarrow \unitinterval$ s.t.~$\sum_a \mu(a) {=} 1$.
The set $\Distr(A)$ contains all distributions over $A$.

\begin{definition}[MDP]
A \emph{Markov decision process (MDP)} is a tuple $M = \mdpT$ with a finite set $S$ of states, an initial state $\sinit \in S$, a finite (indexed) set $\Act$ of actions, and a partial transition function $\mpm \colon S \times \Act \nrightarrow \Distr(S)$.
\end{definition}
For an MDP $M$, we define the \emph{available actions} in  $s \in S$ as
$\Act(s) \coloneqq$ $ \{ \act \in \Act \mid \mpm(s,\act) \neq \bot \}$;
we denote $\mpm(s,\act,s') \coloneqq \mpm(s,\act)(s')$.
An MDP with $|\Act(s)|=1$ for each $s \in S$ is a \emph{Markov chain (MC)};
we denote MCs as tuples $\mcT$.
A (deterministic, memoryless) \emph{policy} is a function $\sched \colon S \rightarrow \Act$ with $\sched(s) \in \Act(s)$ for all $s \in S$. The set $\schedulers$ denotes the policies for MDP $M$.
A~policy $\sched \in \schedulers$ induces the MC~$\imc = \left(S,s_0,\mpm^{\sched} \right)$ where \mbox{$\mpm^{\sched}(s) = \mpm(s,\sched(s))$}.
A \emph{partial} (deterministic, memoryless) policy is a function $\sched \colon S \nrightarrow \Act$.

\paragraph{Specifications.}
We consider indefinite-horizon reachability and expected reward properties as well as discounted (infinite-horizon) total reward objectives~\cite{puterman2014markov}. To simplify the exposition, we formalize our approach only for the \emph{maximal reachability probability}\footnote{Our implementation supports all the aforementioned specifications.}.
Formally, let $M=\mcT$ be an MC, and let $G \subseteq S$ be a set of \emph{goal states}.
Let $\prob[M,s]{\F{G}}$ denote the probability of reaching (some state in) $G$ from state $s \in S$.
Let $\mathbb{P}[M \models \F{G}]$ denote $\prob[M,\sinit]{\F{G}}$.
For MDPs, specifications are taken over the best and worst possible resolution of the non-determinism.
Assume MDP $M = \mdpT$.
The maximal reachability probability of $G$ from state $s_0$ in $M$ is 
$\mathbb{P}_{\max}[M \models \F{G}] \coloneqq \sup_{\sched \in \schedulers} \mathbb{P}[\imc \models \F{G}]$.
We denote $V(\sched) \coloneqq \mathbb{P}[\imc \models \F{G}]$ as the \emph{value} of policy \sched.
An optimal policy that maximizes the value is denoted with \schedopt. The value of MDP is defined as $V(M) = V(\schedopt)$.
Details are given in~\cite{Baier2018}.

\subsection{Representing Policies as Decision Trees}
\label{sec:dt}

\paragraph{Symbolic MDPs.}
We aim to represent policies in MDPs symbolically.
Inspired by the PRISM language~\cite{KNP11}, we assume a finite set of bounded integer variables \variables,
and thus states are mappings $s \colon \variables \rightarrow \integers$.
\emph{State predicates} are inequalities of the form $v \leq b$ with $v \in \variables$ and $b \in \integers$; the set of such predicates is denoted $\predicates$.
A state $s$ \emph{satisfies a predicate} $v \leq b$ iff $s(v) \leq b$; we denote this with $s \models (v \leq b)$.

\paragraph{The random action. }
To concisely represent policies, it is convenient to allow a policy to take some dedicated \emph{arbitrary} action. We explicitly create this arbitrary action  $\actrandom$ for every state which uniformly selects one of the (available) actions. Formally, we define $M' = (S,\sinit,\Act',\mpm')$ with $\Act'(s)=\Act(s)\cup\{\actrandom\}$, $\mpm'(s,\actrandom,s') = \frac{1}{|\Act(s)|}\sum_{\act \in \Act(s)}\mpm(s,\act,s')$. 
Henceforth, we assume that MDP $M' = M$, i.e., that every MDP contains an action $\actrandom$.
This is sound for the specifications considered in this paper: it holds that the optimal value is achievable by a memoryless deterministic policy $\sigma^{*}$ as per~\cite[Thm 7.9.1]{Put94}. 
\begin{corollary} \emph{(Proof in App.~\ref{appendix:proof:random})}
\label{corollary:random}
Given an MDP $M$, it holds that $V(M)~=~V(M')$. % for all specifications considered in this paper.
\end{corollary}
Any policy in MDP $M'$ can be mimicked by a (randomizing) policy in  $M$.

\begin{remark}[Interpretability of random action]
Adding the random action makes it explicit that a policy may randomly pick either available action. Sometimes, having this opportunity makes for more concise policies. A possible downside is that the policy in $M'$ may not reflect a memoryless deterministic policy in $M$. 

\end{remark}

\paragraph{Trees.}
A (binary) \emph{tree} is a tuple 
$T = (n_0,N,L,l,r)$  with the \emph{root node} $n_0$, the set~$N$ of \emph{inner nodes}, the set $L$ of \emph{leaf nodes}, and functions $l,r \colon N \rightarrow N\cup L $ defining the \emph{left} and \emph{right successors} of the inner nodes, respectively. 
A \emph{path} (of length $k$) in a tree~$T$ is a sequence $\pi = n_0\ldots n_k$ of nodes s.t.\ $\forall 0 {<} i {\leq} k: n_i \in \{l(n_{i-1}),r(n_{i-1})\}$.
Path $\pi$ is \emph{complete} if it ends in a leaf node.
The \emph{depth} of $T$ is the length of its longest path.
 
\begin{definition}[Decision tree]
Assume an MDP $M = \mdpT$ defined over the set \variables.
A \emph{decision tree} (DT) for $M$ is a tuple $\tree = \treeT$ where (i)~$T$ is a binary tree, (ii)~\emph{predicate function} $\gamma \colon N \rightarrow \predicates$ assigns to inner nodes a state predicate, and (iii)~\emph{action function} $\delta \colon L \rightarrow \Act$ assigns to leaf nodes an action. 
\end{definition}
We lift the notions of inner and leaf nodes, paths and depths of trees to DTs. DTs of depth $k$ are further denoted as $k$-DTs.

\begin{definition}[Corresponding states]
The set $\sts{\tree}{n}$ of states that corresponds to a node $n$ is recursively defined as follows: $\sts{\tree}{n_0} = S$, $\sts{\tree}{l(n)} = \{ s \in \sts{\tree}{n} \mid s \models \gamma(n) \}$, and $\sts{\tree}{r(n)} = \{ s \in \sts{\tree}{n} \mid s \not\models \gamma(n) \}$.
\end{definition}
Note that the sets $\{ \sts{\tree}{n} \mid n \in L\}$ represent a partition of $S$. Thus, we can define $\leaf(s)$ as the unique leaf node $n \in L$ such that $s \in \sts{\tree}{n}$.

\begin{definition}[Induced policy]
\label{def:induced-policy}
The DT $\tree$ \emph{induces policy} $\sched[\tree] \colon S \rightarrow \Act$ with $\sched[\tree](s) = \delta(\leaf(s))$ if $\delta(\leaf(s)) \in Act(s)$ and $\sched[\tree](s) = \actrandom$ otherwise. The \emph{value} $V(\tree)$ of the DT \tree is defined as the value of $\sched[\tree]$.
\end{definition}

\begin{example}
Consider an MDP presented in Fig.~\ref{fig:example:maze} and a DT $\tree_c$ depicted in Fig.~\ref{fig:example:maze:tree-2}.
Assume a state $s \in S$ with $s(x)=4$ and $s(y)=3$.
Then, $\leaf(s)$ is the rightmost node of $\tree_c$ and thus $\tree_c$ induces policy $\sched[{\tree_c}]$ where $\sched[{\tree_c}](s) = up$.
\end{example}

\begin{remark}[Fallback action interpretability]
    Our approach can synthesize a DT that assigns an action that is not available at a given state, in which case we use the random action as a fallback. This setup avoids that the DT must precisely capture when an action is available and allows for smaller DTs.
    Note that the information 
    which actions are available is usually also accessible in another format (e.g. masks or shields heavily used in reinforcement learning settings). 
\end{remark}

\subsection{Problem Statement}
This paper's key problem is to find a decision tree with maximum value (e.g. reachability probability) among all decision trees up to the given depth. 
\begin{mdframed}
\textbf{Bounded-depth synthesis problem}: Given MDP $M$ and bound $k$, find a DT \tree of depth up to $k$ with maximum value.
\end{mdframed}
Akin to~\cite{vos2023optimal}, we are interested in anytime synthesis: the faster we find a DT that achieves a high value, the better. This is a variant of widely studied a \emph{policy mapping problem}~\cite{gupta2015policy,brazdil2015counterexample,likmeta2020combining,topin2021iterative,ashok2021dtcontrol}: find a DT inducing a given fixed policy (typically a pre-computed policy that maximizes reachability). In contrast, in the bounded-depth synthesis problem, we seek the best DT up to the given depth.

\section{Solving the fixed-tree policy mapping problem}
\label{sec:policy-mapping}
In this section, we construct an SMT-based subroutine that solves a \emph{fixed-tree policy mapping problem}, a variant of the policy mapping that searches over a set of DTs having the same topology. Additionally, we consider the construction of a witness that explains why a policy cannot be represented as such a DT.
We start with \emph{tree templates} representing a blueprint of the set of~DTs.

\begin{definition}[Tree template]
\label{def:template}
Assume an MDP $M = \mdpT$ defined over the set \variables.
A \emph{tree template}~for $M$ is a tuple $\template = \templateT$ where (i)~$T$ is a binary tree, (ii)~$\Gamma \colon N \rightarrow 2^{\predicates}$ assigns to inner nodes a set of state predicates, and (iii)~$\Delta \colon L \rightarrow {2^{\Act}}$ assigns to leaf nodes subset of actions. 
\end{definition}
A tree template encodes (or instantiates) a set of DTs and, thus, a set of policies.
\vspace{-1em}
\begin{definition}[Instantitation]
Template $\template = \templateT$ \emph{instantiates} DT $\tree = \treeT$, written $\tree \in \template$, if $\gamma(n) \in \Gamma(n)$ for all $n \in N$ and $\delta(n) \in \Delta(n)$ for all $n \in L$.
A policy~\sched is \emph{contained} in template \template, written $\sched \in \template$, if there exists a tree $\tree \in \template$ that induces~\sched; in such a case, we say that \sched is \emph{\template-implementable} (via \tree).
\end{definition}

\vspace{-1em}
\subsection{SMT encoding}
\label{sec:smt}

Let's assume a fixed tree template \template.
In this subsection, we present an SMT formula over the theory of quantifier-free linear integer arithmetic that is satisfiable iff a (partial) policy \sched is \template-implementable.
To this end, we first reparameterize the decision trees and templates in terms of integer-valued variables.

\paragraph{Parameterization.}
We assume that the set of state variables $\variables = \{v_i\}_{i=1}^{|\variables|}$ is ordered and the actions are integers $\Act \subseteq \integers$ and are ordered.
Observe that any state predicate $v_i \leq b$ can be described by a pair of integers $i,b$.
Therefore, the predicate function $\gamma \colon N \rightarrow \predicates$ of a DT can be expressed using two integer-valued functions:
the \emph{variable selection function} $\decisionfun\colon N \rightarrow \{1, \dots, |\variables|\}$
and the \emph{bound selection function} $\boundfun \colon N \rightarrow \integers$.
Then, the predicate function $\gamma_{\decisionfun, \boundfun}$ associated with $\decisionfun,\boundfun$ is defined as $\gamma_{\decisionfun, \boundfun}(n) \coloneqq ( v_{\decisionfun(n)} \leq \boundfun(n))$.
Likewise, with the \emph{action selection function} $\actionfun \colon L \rightarrow \Act$, we can define action function $\delta_\actionfun$ where $\delta_\actionfun(n) \coloneqq \actionfun(n)$.

\begin{definition}[Parameterization]
\label{def:parameterization}
Given \template and functions \decisionfun, \boundfun and \actionfun as above. The \emph{parameterization} $f = (\decisionfun, \boundfun, \actionfun)$ is \emph{contained} in \template if the DT $(T, \gamma_{\decisionfun,\boundfun}, \delta_\actionfun) \in \template$; we write $\template(f)$ to denote this DT. The value $V(f)$ of $f$ is the value $V(\template(f))$.
\end{definition}
Towards an SMT encoding, we introduce the variables $\params = \{ \boundfun_n, \decisionfun_n \mid n \in N \} \cup \{ \actionfun_n \mid n \in L \}$ such that $\boundfun_n$ encodes the value of $\boundfun(n)$, etc.
We write $f(x)$ to denote the value of any variable $x \in \params$.

While parameterization $f$ allows us to generate a DT $\template(f)$ from a template, a \emph{set of parameterizations} allows us to generate restrictions of the template.
 
\begin{definition}[Parameterization set]
\label{def:family}
Given a template \template, a  (rectangular) \emph{set of parameterizations} is a function $\family \colon \params \rightarrow 2^{\integers}$ where, for any $n \in N$, it holds that $\family(\decisionfun_n) \subseteq \{1, \dots, |\variables|\}$, $\family(\boundfun_n) \subseteq \integers$ and for any $n \in L$, $\family(\actionfun_n) \subseteq \Act$.
Parameterization $f$ belongs to set \family, denoted $f \in \family$, if $f(x) \in \family(x)$ for any variable $x \in \params$. We write \ftemplate
for the template that contains exactly the DTs $\template(f)$ with $f \in \family$.
\end{definition}

\paragraph{Encoding.}
We present an encoding that checks for a given \family whether a (partial) policy \sched is \ftemplate-implementable.
In the encoding below, we assume that \sched is a partial policy where $\sched(s) = \bot$ if $s \in G$ or $s$ is unreachable in \imc.
The formula $\fpol(\family, \sched)$ uses the variables from $\params$ and abstractly is a conjunction of
(1)~$\fdom(\family)$, which limits possible parameterizations to \family, and
(2)~$\fact_{s,\sched(s)}$, which ensures that the tree represents \sched in every state.
\[ \fpol(\family, \sched) \coloneqq\quad \fdom(\family) \land \bigwedge_{s \in S, \sched(s) \neq \bot} \fact_{s,\sched(s)}.   \]
We explain the individual constraints below.

\paragraph{Contained in a template.}
$\fdom(\family)$ is a constraint that limits the variables to ensure that a DT is contained in \ftemplate

\[ \fdom(\family) \coloneqq \bigwedge\nolimits_{x \in \params} \, x \in \family(x) \]

\paragraph{State-representing leaf.}
Before we discuss expressing the correct policy, we match the leaves of the DT with states. 
Recall that a DT partitions the state space $S$ into $\{ \sts{\tree}{n} \mid n \in L\}$. In a tree template, this partitioning depends on the choice of $f$, i.e.~on the choice of \boundfun and \decisionfun. 
We create a formula $\fsel_{s,n}(f)$ that is true iff $s \in \sts{\template(f)}{n}$. 
First, for the unique path $n_0,\dots,n_k$ from root to leaf $n = n_k$, we define the relation $\meetsin$ such that
$\meetsin\, \coloneqq {\leq}$ if $n_{i+1} = l(n_i)$ 
and $\meetsin\, \coloneqq {>}$ if $n_{i+1} = r(n_i)$.
The formula $\fsel_{s,n}$ is then given by:
\begin{align*} \fsel_{s,n} \coloneqq\quad  \textstyle
 &\bigwedge\textstyle_{i=0}^{k-1} \bigwedge\textstyle_{j=1}^{|\variables|} \quad \big(j 
 = \decisionfun_{n_i}\big) \rightarrow \big(s(v_j)\meetsin \boundfun_{n_i}\big). 
\end{align*}

\paragraph{Action-choice} 
For the state $s$ to take action \act in a DT, it must hold that the (unique) leaf $n$ that represents state $s$ picks this action \act.
The randomized action \actrandom can also be picked when this leaf selects an unavailable action.

\vspace{-1em}
\begin{align*}
\fact_{s,\act} & \coloneqq \quad  \bigwedge_{n \in L}  \fsel_{s,n} \rightarrow \fact_{s,\act,n}, \qquad \text{using} \\
\fact_{s,\act,n} & \coloneqq
\begin{dcases*}
\actionfun_{n} = \act & if $\act \neq \actrandom$ \\
\actionfun_{n} \in \{\actrandom\} \cup \Act {\setminus} \Act(s)  & otherwise \\
\end{dcases*}    
\end{align*}

\paragraph{Correctness.}
A parameterization~$f$ satisfying $\fpol(\family,\sched)$ is denoted as $f \models \fpol(\family,\sched)$.

\begin{theorem} \emph{(Proof in App.~\ref{appendix:proof:encoding}.)}
\label{thm:smt}
Assume a parameterization $f$, a set of parameterizations \family, and a policy $\sched \in \schedulers$.
Then $f \models \fpol(\family,\sched)$ iff $f \in \family$ and forall $s \in S, \sched(s) \neq \bot $ implies $ \sched(s) = \sched[\template(f)]$.
\end{theorem}

\vspace{-1.5em}
\subsection{Diagnosing non-implementable policies.}
%\vspace{-1em}
\label{sec:harmonization}
\subsubsection{Unsatisfiable cores.}
For further usage, we want to examine why $\fpol(\family,\sched)$ is unsatisfiable. The conjunction $\fpol(\family,\sched)$ is essentially a set of domain constraints and a set of action-choice constraints. 
An \emph{unsatisfiable core}~\cite{biere2009handbook} is a natural way of explaining why $\fpol(\family,\sched)$ cannot be satisfied by using a subset of the constraints.
The computation of unsatisfiable cores is supported by existing SMT solvers, e.g.~by Z3~\cite{z3}. An unsatisfiable core $\unsatcore$ is simply a subset of parameters and state-leaf constraints that make $\fpol(\family,\sched)$ unsatisfiable.

\begin{definition}[Unsatisfiable core]
\label{def:unsat-core}
An unsatisfiable core for \family and \sched is a pair $\unsatcore = (\paramsuc \subseteq \params,\statepathuc \subseteq S \times L)$ s.t.\ the following formula is unsatisfiable
\[ \fpol^{\unsatcore}(\family,\sched) \coloneqq \bigwedge_{x \in \paramsuc} x \in \family(x) \land \bigwedge_{s,n \in \statepathuc}\fact_{s,\sched(s), n}.
%\vspace{-1em}
\]
The set of \emph{critical states} for \family and \sched is $\criticalstates \coloneqq \{ s \in S \mid \exists n: (s,n) \in \statepathuc \}$.
\end{definition}
In our experiments, \statepathuc is often much smaller than $S \times L$: In fact, it is not uncommon that \statepathuc contains only two elements.

Given a subset $S' \subseteq S$ and a policy \sched, we define the \emph{restriction}~$\sched[\downarrow S']$ of \sched on $S'$ as $\sched[\downarrow S'](s) = \sched(s)$ if $s \in S'$ and $\sched[\downarrow S'](s)=\bot$ otherwise.
The following proposition asserts that the set of critical states witnesses the non-implementability of~\sched.
\begin{proposition}
Let \criticalstates be critical states for \family and \sched.  The partial policy \criticalsched is not \ftemplate-implementable. 
\end{proposition}

\begin{remark}
\label{remark:custom-unsat-core}
$\fpol(\family,\sched)$ can have multiple unsatisfiable cores and thus multiple sets of critical states. We prefer critical states where the policy choice also critically affects the value, similar to~\cite{brazdil2015counterexample}. 
Assume that the state space $S =\{s_i\}_{i=1}^{|S|}$ is ordered.
Instead of deducing critical states from $\fpol(\family,\sched)$, we deduce them from $\fdom(\family) \land \bigwedge_{i = 1}^m \fact_{s_i,\sched(s_i)}$ where $m$ is the smallest state index that makes the formula unsatisfiable. 
We use a breadth-first search from the initial state\footnote{We tested various other orders, but this one seemed to be the most robust. 
}.
\end{remark}

\vspace{-0.5em}
\subsubsection{Harmonising parameterizations.}
The set of critical states \criticalstates identifies which state decisions $\sched(s)$, 
$s \in \criticalstates$, are incompatible with \ftemplate, i.e. make \criticalsched not \ftemplate-implementable.
To gain another insight into this incompatibility, we reverse the question and ask what makes \ftemplate incompatible with \criticalsched.
Since there is no one parameterization $f \in \family$ s.t. $\sched[\template(f)] = \sched$, we seek two parameterizations $f_1,f_2$ that differ in only one variable assignment and where policies $\sched[\template(f_1)],\sched[\template(f_2)]$ \emph{together} describe \sched on the set \criticalstates.
We now formalize this concept.

\begin{definition}[Harmonizing parameterizations]
\label{def:harmonization}
Let \sched be a policy that is not \ftemplate-implementable. Parameterizations $f_1,f_2 \in \family$ are called \emph{harmonizing} for \sched if $\exists!\, x \in \params : f_1(x) \neq f_2(x)$ and for any state $s \in \criticalstates$, $\sched(s) \in \{\sched[\template(f_1)](s),\sched[\template(f_2)](s)\}$. The variable $x$ above is \emph{the harmonizing variable}.
\end{definition}

\paragraph{Encoding.} We provide an SMT encoding that yields harmonizing parameterizations. Intuitively, we seek two trees, one from the original tree template \template (encoded with the original variables $\params$), and the other from the identical tree template $\template'$ (encoded with the primed variables $\params[\template']$). 
We first describe the requirement for the two trees to differ in at most one variable, using the fact that $\params$ and $\params[\template']$ are ordered similarly and using an auxiliary integer variable $h$ to encode the index of the harmonizing variable:
\vspace{-.2em}
\[\fharm(\family) \coloneqq \bigwedge\nolimits_{1 \leq i \leq |\params|}\, h \neq i \,\rightarrow\, x_i = x'_i. \]
The primed variables are associated with the identical set $\family'$ of parameterizations: $\family'(x') = \family(x)$. Then, the encoding is given by:
\vspace{-.2em}
\[\fpol^H(\family,\sched) \coloneqq \fdom(\family) \land \fdom(\family') \land \fharm(\family) \land \bigwedge\nolimits_{s \in \criticalstates} \left( \fact_{s, \sched(s)} \lor \fact'_{s,\sched(s)} \right)  \]
using $\fdom$ as above and $\fact'_{s,\sched(s)}$ is identical to $\fact_{s,\sched(s)}$ except every occurrence of variables $\decisionfun_n,\boundfun_n,\actionfun_n$ is replaced with their primed counterpart.

\begin{theorem} \emph{(Proof in App.~\ref{appendix:proof:harmonization}.)}
\label{thm:harmonization}
Assume parameterizations $f_1,f_2$, two (identical) sets of parameterizations $\family,\family'$, a policy $\sched$ that is not $\template(\family)$-impelementable, and an evaluation $f_h \in \integers$ of the variable $h$.
Then $f_h,f_1,f_2 \models \fpol^H(\family,\sched)$ iff $f_1,f_2$ are harmonizing for $\sched$ via a harmonizing variable $x_{f_h} \in \params$.
\end{theorem}

\color{red!50!black}

\color{black}

\vspace{-1em}
\section{Fixed-tree Synthesis Problem}
\label{sec:fixed-tree-synthesis}

We now shift our focus to the fixed-tree synthesis problem: given template~\template, we seek for a DT $\tree \in \template$ with the highest value $V(\tree)$.
We first show the NP-hardness of this problem. Afterwards, we outline our abstraction refinement approach.

\subsection{Problem Complexity}

The complexity of fixed-tree synthesis problem follows from the result of~\cite{laurent1976constructing}, which showed that finding optimal decision trees (in terms of size) is NP-complete.

\begin{theorem} \emph{(Proof in App.~\ref{appendix:thm:np-hardness}.)}
\label{thm:np-hardness}
The decision variant of the fixed-tree synthesis problem is NP-hard.
\end{theorem}

Our proof follows similar principles to the proof in~\cite{laurent1976constructing}. Here we provide the main idea of the reduction, for full proof refer to App.~\ref{appendix:thm:np-hardness}. The proof is a reduction from the Exact Cover by 3-sets (X3C) problem to the problem of deciding whether there exists a decision tree of depth $k$ for MDP $M$ that implements a policy that reaches a goal state with probability above $0.5$. The X3C problem is a known NP-Complete problem which asks whether for some set $U$ and set $T = \{T_1,T_2,\dots,T_j\}$ where $|T_{i}|=3$ and $T_{i} \subseteq U$ there exists an exact cover $T' \subseteq T$. The main idea is to create an MDP where the states correspond to the elements of the set $U$ for X3C, such that in every state of the MDP exactly one unique action progresses towards the goal state $n+4$ while the rest of the actions lead to a sink state as shown in Fig.~\ref{fig:x3c-to-mdp}. 
This means that a DT that reaches a goal state with probability above 0.5 needs to distinguish all states of the MDP. 
There are two types of state variables in this MDP: 1) those that correspond to the sets in $T$ which have a value of 1 when $x \in T_i$ and 2) variables for each state which have a value of 1 only in one unique state. In the App.~\ref{appendix:thm:np-hardness} we show that for such an MDP there exists a DT of depth $k = |U|/3+2$ that reaches the goal state if and only if there exists an exact cover $T'$.

\begin{figure}[t]
    \centering
    \includegraphics[width=0.6\textwidth]{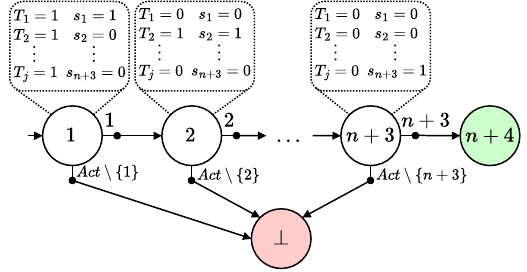}
    \vspace{-1em}
    \caption{MDP construction for the reduction from X3C problem.}
    \label{fig:x3c-to-mdp}
\end{figure}

\subsection{Abstraction Refinement}
\label{sec:refinement}

Due to the finiteness of the state space $S$ and the variable set \variables, the number of variable valuations is finite as well.
Thus, we consider a finite number of bound selection functions: $\boundfun(n) \in \{s(v) \mid s \in S, v \in \variables\}$.
We use \superfamily to denote the largest finite set of such parameterizations for \template.
The following definition introduces an abstraction that overapproximates the set of policies induced by described by \family.

\begin{definition}[Family-MDP]
\label{def:family-mdp}
Assume an MDP $M = \left(S,\sinit,\Act^M,P\right)$, a tree template \template and a family $\family \subseteq \superfamily$. An \family-MDP is an MDP $M(\family) = \left(S,\sinit,\Act,P'\right)$ where $P'(s,\act) = P(s, \act)$ if $\exists f \in \family: \sched[\template(f)](s) = \act$ and $P(s,\act) = \bot$ otherwise.
\end{definition}

An \family-MDP $M(\family)$ is a sub-MDP of $M$ 
where action \act is enabled in state $s$ only if at least some DT $\template(f)$, $f\in \family$, induces a policy that associates $s$ with \act.
We create this family-MDP by checking the satisfiability of $\fdom(\family) \land \fact_{s,\act}$ for every state-action pair.
The following proposition asserts that the set of policies for $M(\family)$ includes all policies obtained from \template using assignments from \family. 
\begin{proposition}
\label{proposition:family-mdp}
Let $\family \subseteq \superfamily$. Then $\forall f \in \family: \sched[\template(f)] \in \Sigma^{M(\family)}$.
\end{proposition}
We stress that $M(\family)$ is an abstraction, i.e.~$\Sigma^{M(\family)}$ may contain policies that are not \ftemplate-implementable.

\setlength{\textfloatsep}{10pt}
\begin{algorithm}[t]
\DontPrintSemicolon
\SetKwFunction{findRobust}{\textsc{findRobust}}
\SetKwFunction{testUnsat}{\textsc{testUnsat}}
\SetKwFunction{buildTree}{\textsc{buildTree}}
\SetKwInOut{Input}{Input}
\SetKwInOut{Output}{Output}
\SetKw{Continue}{continue}
\SetKw{Yield}{yield}
\SetKwComment{Comment}{$\triangleright$\ }{}
\Input{MDP $M$, goal set $G$, tree template \template}
\Output{DT $\template(f)$ s.t.~$f \in \mathrm{arg\,max}_{f \in \superfamily}V(f)$ }
\BlankLine
\SetKwProg{myfunction}{Function}{}{}
$\mathfrak{F} \gets \textsc{stack()}, \mathfrak{F}.\textsc{push}(\superfamily), f_{best} = \varnothing, V_{best} = -\infty$\;
\label{alg:refinement:start}
\While{$\mathfrak{F} \neq \emptyset$} {
$\family \gets \mathfrak{F}.\textsc{pop}()$\;
\label{alg:refinement:pop}
$M(\family) \gets \textsc{buildFamilyMdp}(M,\template,\family)$ \Comment*[r]{applying Def.~\ref{def:family-mdp}}
$\sched,V(\sched) \gets \textsc{modelCheck}(\probmax[M(\family)]{\F{G}})$\;
\label{alg:refinement:solve}
\lIf{$V(\sched)\leq V_{best}$ }{\Continue}
\label{alg:refinement:prune}
\lIf{$f \models \fpol(\family,\sched)$}{
    $f_{best} \gets f,\, V_{best} \gets V(\sched)$, 
    \Continue
}
\label{alg:refinement:sat}
$\paramsuc,\statepathuc \gets \textsc{unsatCore}(\fpol(\family,\sched))$ 
\Comment*[r]{ applying Remark~\ref{remark:custom-unsat-core} and  Def.~\ref{def:unsat-core}}
\label{alg:refinement:unsatcore}
\If(\Comment*[f]{according to~Sec.~\ref{sec:harmonization}}){$f_h,f_1,f_2 \models \fpol^H(\family,\sched)$} 
{
    \For{$f \in \{f_1,f_2\}$} {
    \label{alg:try-f1-f2-1}
        \lIf{$V(f) > V_{best}$}{
            $f_{best} \gets f$, $V_{best} \gets V(f)$
        \label{alg:try-f1-f2-2}
        }
    }
    $\family_1,\family_2 \gets \textsc{splitInformed}(\family,f_h,f_1,f_2)$\;
} \lElse {
    $\family_1,\family_2 \gets \textsc{splitArbitrary}(\family,\paramsuc)$
}
\label{alg:refinement:ref-end}
$\mathfrak{F}.\textsc{push}(\family_1), \ \mathfrak{F}.\textsc{push}(\family_2)$\;
\label{alg:refinement:push}
}
\Return $\template(f)$ \Comment*[r]{applying Def.~\ref{def:parameterization}}
%}
\caption{Recursive DT construction}
\label{alg:refinement}
\end{algorithm}

We propose an abstraction-refinement-based approach to synthesize a tree from a given template that maximizes the value among all trees in the template.  
The basic idea is borrowed from~\cite{cegar}. For a given family \family (starting from \superfamily), an abstraction $M(\family)$ is built that allows either prune the family from the search space or splitting the family into smaller subfamilies that are recursively analyzed.
In order to prune family~\family, we compute the maximizing policy \sched for (sub-)MDP $M(\family)$ and either show that no $f \in \family$ has a better value than the current optimum or that \sched is \template-implementable using SMT encoding from above, updating the optimum. Otherwise, we will divide and conquer \family by partitioning it. We guide the partitioning using harmonization.

The approach is summarized in Algorithm~\ref{alg:refinement}. On l.\ref{alg:refinement:start}, we start with a stack $\mathfrak{F}$ of sub-families, initially containing \superfamily, and initialize the running optimum $f_{best}$ and its value $V_{best}$.
In every iteration, we pop a family \family from the stack, build the corresponding $M(\family)$ and compute the policy \sched that maximizes $\prob[M(\family)]{\F{G}}$ (ll.~\ref{alg:refinement:pop}-\ref{alg:refinement:solve}).
If its value $V(\sched)$ is worse than $V_{best}$, then no assignment in \family induces a tree with a better value, and thus this family is pruned from the search space (l.\ref{alg:refinement:prune}).
Otherwise, we solve the SMT formula $\fpol(\family,\sched)$ to check whether \sched is \template-implementable.
If $f \in \family$ is a parameter assignment s.t.~$\sched[\template(f)]=\sched$, we update the running optimum and prune the family (l.\ref{alg:refinement:sat}).
Otherwise, on ll.\ref{alg:refinement:unsatcore}-\ref{alg:refinement:ref-end} we split~\family into sub-families $\family_1,\family_2$ and push these onto the stack $\mathfrak{F}$ (l.\ref{alg:refinement:push}).

\begin{theorem}
\label{thm:sound}
Algorithm~\ref{alg:refinement} is sound and complete.
\emph{Proof in App.~\ref{appendix:thm:sound:proof}.}
\end{theorem}

Any nontrivial splitting makes Alg.~\ref{alg:refinement} terminate: the number $|\superfamily|$ of parameterizations is finite and, in the worst case, a nontrivial splitting yields a family with a single assignment~$f$, in which case $M(\{f\})$ is an MC with only one policy $\sched[\template(f)]$ and the SMT formula $\fpol(\{f\},\sched[\template(f)])$ is satisfiable with parameter assignment $f$.
However, even for a tree template of small depth, the number of template instantiations is insurmountable, and thus, a proper splitting strategy should yield sub-families that can be pruned as soon as possible.

To deal with enormous design spaces, the abstraction refinement framework of~\cite{cegar} successfully used the notion of inconsistent variables (holes), where a split was made on a hole for which the optimizing policy wanted to pick multiple values.
In our framework, the harmonising variable $x_h$ plays the role of this inconsistent hole, and therefore, on l.\ref{alg:refinement:unsatcore}, we extract the unsatisfiable core (see Remark~\ref{remark:custom-unsat-core}) and solve the SMT formula $\fpol^H(\family,\sched)$, as described in Sec.~\ref{sec:harmonization}.
Assume the formula is SAT with the harmonising variable $x_h$ and $f_1$,$f_2$ is the corresponding pair of trees that differ in the value of $x_h$.
We then split \family into subfamilies $\family_1,\family_2$ by splitting the domain $\family(x_h)$ s.t.~$f_1(x_h) \in \family_1(x_h)$ and $f_2(x_h) \in \family_2(x_h)$; if $x_h$ is a variable encoding bound selection function $\boundfun_n$, then its domain (initially, an interval) is split into two sub-intervals.
The idea here is to build in subsequent iterations sub-MDPs $M(\family_1),M(\family_2)$ that do not contain~\sched, although this cannot always be guaranteed. Otherwise, if the harmonising formula is unsatisfiable, we split the family arbitrarily on some parameter from $\paramsuc$. We remark that during our experiments, the harmonising formula was practically always satisfiable.
Additionally, on ll.\ref{alg:try-f1-f2-1}-\ref{alg:try-f1-f2-2} we update the value of $V_{best}$ based on the values $V(f_1),V(f_2)$. Empirically, this leads to a mildly better performance.

\subsubsection{Bounded-depth synthesis.}
Even for modest values of~$k$, Algorithm~\ref{alg:refinement} cannot explore all parameter assignments. Finding good assignments early can accelerate abstraction refinement~\cite{cegar,andriushchenko2021inductive} as it prunes sub-optimal families faster. Thus, when searching for the optimal $k$-tree, it can be beneficial to first go through families of $0$-, $1$-, $2$-trees, etc., where good values are easier to find.
This idea inspired the \emph{bounded-depth mode} of our abstraction refinement approach that proceeds as follows.
We iteratively use Algorithm~\ref{alg:refinement} on templates of trees of depths $0,1,\dots,k_{\max}$; in each iteration, we keep the current optimum $f_{best}$ and its value $V_{best}$ and use it in subsequent iterations.
To ensure that the algorithm reaches depth $k_{\max}$, we run Algorithm~\ref{alg:refinement} on lower depths $0,1,\dots,k_{\max}{-}1$ with a timeout, that we empirically choose to be $t/(2 \cdot k_{\max})$, such that at least 50\% of the given time is dedicated to the search on depth $k_{\max}$.

\subsubsection{Tree hints.}
Having (partially) explored a family of $(\kmone)$-trees, we can accelerate the search for the best $k$-tree even further by first looking at $k$-trees that mimic a good $(\kmone)$-tree $\tree_\kmone$. Our abstraction-refinement approach on families naturally supports this idea: before running Algorithm~\ref{alg:refinement} with a stack containing the whole family $\family^{\template_k}$ for a tree template $\template_k$ of depth $k$, we can first make it look within the family $\family'$ of assignments that mimic $\tree_\kmone$. Intuitively, $\family' \subset \family^{\template_k}$ describes all $k$-trees
that, in the inner nodes on the first $k-1$ levels, behave according to $\tree_\kmone$, and behave arbitrarily on the last $k$-th level as well as in the leaves.
Putting family $\family'$ on top of the stack prioritizes the search for the best $k$-tree within this family, increasing the chance of finding good $k$-trees early.

\subsubsection{Tree post-processing.}
\label{sec:postprocessing}
We perform two steps to remove redundant nodes.
First, we remove every node $n \in N \cup L$ for which $\sts{\tree}{n} = \emptyset$, that is, no state $s \in S$ can take a path to $n$.
Second, if for a node $n \in N$ it holds that $l(n),r(n) \in L$ and $\delta(l(n)) = \delta(r(n))$, i.e.~the children are leaves selecting the same action, then both children are removed, and $n$ becomes a leaf associated with this action.

\section{Experimental Evaluation}
\label{sec:experiments}

In the experimental evaluation, we investigate the performance of the proposed synthesis algorithm \dtpaynt. \dtpaynt is an algorithm that solves the bounded-depth synthesis problem (with an explicit timeout) using abstraction refinement in a bounded-depth mode using tree hints as described in Sec.~\ref{sec:refinement}. \dtpaynt is implemented on top of \paynt~\cite{DBLP:conf/cav/AndriushchenkoC21} and \storm~\cite{STORM}, utilising Z3~\cite{z3} to solve SMT queries. The implementation and all the considered benchmarks are publicly available\footnote{\url{https://doi.org/10.5281/zenodo.15228774}}.
Our evaluation focuses on the following four questions:
\begin{enumerate}[topsep=2pt,leftmargin=2.5em]
    \item[Q1:] \emph{Does \dtpaynt outperform \omdt~\cite{vos2023optimal} that is based on a MILP formulation? Can \dtpaynt scale to more complex MDPs?}
    \item[Q2:] \emph{Does \dtpaynt provide a good trade-off between value and size of the synthesized trees compared to \dtcontrol~\cite{ashok2021dtcontrol}?} 
    \item[Q3:]  \emph{Can \dtpaynt handle huge MDPs having up to 1M states?} 
    \item[Q4:] \emph{Can we use \dtpaynt as a reduction procedure for large DTs?}
\end{enumerate}

\paragraph{Setting.}
The timeout (TO) for all experiments was 20 minutes\footnote{Based on our preliminary experiments as well as the results from~\cite{vos2023optimal}.
}.
All the experiments were run on a machine equipped with AMD EPYC 9124 16-core Processor and 380GB RAM. Each method was run on a single CPU core..

\paragraph{Benchmarks.}
In order to answer the questions Q1 and Q2, we consider three types of benchmarks: (1)~The 11 models from \omdt~\cite{vos2023optimal}
using expected discounted rewards. (2)~The standard MDP benchmarks from the QComp evaluation~\cite{budde2020correctness} with 8 models with state variables defined.
{Since \omdt requires discounted-reward specifications, we derived such specifications for these QComp models.
(3) Two fully observable variants of the classical \emph{maze}~\cite{hauskrecht1997incremental} problem with a discounted-reward specifications. 
Information about the models is reported in Appendix~\ref{appendix:model-info}.
In order to perform a fair comparison of the tools, we modify the models as follows: (1) We equip all the models with the action $\actrandom$ that uniformly selects the available actions (recall Def.~\ref{def:induced-policy}). (2) We ensure that in every state, all actions in $\Act$ are available\footnote{The actions added to the original QComp models behave as $\actrandom$.}, as it is required by the available implementation of \omdt. We note that, in theory, the approach proposed in~\cite{vos2023optimal} can handle MDPs where not every action is available in every state, however, there is currently no implementation which supports this. Also note that this modification effectively makes the synthesis task on these models more difficult. We group the benchmarks into two categories: (1) \emph{smaller models} with the number of choices ($|S| \cdot |Act|$)  below 3k and (2) \emph{larger models} with up to 10k states and 175k choices.
Overall, we have 10 smaller models and 11 larger~models.

\vspace{-0.5em}
\subsection*{Q1: Comparison with \omdt} 

We consider the bounded-depth synthesis problem and compare \dtpaynt with \omdt.
We performed a comparison including all benchmarks and tree depths up to $k=8$.
The comparison therefore includes $21\cdot 8=168$ instances.

\setlength{\textfloatsep}{12pt}
\begin{figure}[t]
    \centering
    \begin{subfigure}[t]{0.48\textwidth}
        \centering
        \resizebox{\textwidth}{!}{\input{figures/plot-q1-scatter-smaller-models.pgf}}
        %\caption{}
        \label{fig:omdt-scatter:depth-le3}
    \end{subfigure}
    \hfill
    \begin{subfigure}[t]{0.48\textwidth}
        \centering
        \resizebox{1\textwidth}{!}{\input{figures/plot-q1-scatter-larger-models.pgf}}
       % \caption{}
        \label{fig:omdt-scatter:depth-g3}
    \end{subfigure}
    \vspace{-2em}
    \caption{Comparison on the bounded-depth synthesis problem for Q1. The scatter plot shows the normalized values of the best $k$-DTs found in the 20-minute timeout. A point $(x,y)$ shows
    the value of the best $k$-DT found by \dtpaynt (the $x$-value) and \omdt (the $y$-value) for a specific model and a specific depth~$k$. Points below the diagonal shows the synthesis problems where \dtpaynt outperforms \omdt.
    The plot on the left compares performance on smaller models (less than 3k choices), and the plot on the right compares performance on larger models.
    }
    \label{fig:omdt-scatter}
   
\end{figure}

\paragraph{Results.}
Figure~\ref{fig:omdt-scatter} shows two scatter plots that visualize the normalized values of the best $k$-DTs found by \dtpaynt and \omdt. The values are normalized such that 0 corresponds to a uniform random policy and 1 to an optimal MDP policy. The uniform random policy can be represented with a 0-DT that chooses the action $\actrandom$\footnote{Note that this value is not a strict lower bound on the worst policy.
}.  Note that the normalized value 1.0 represents the value of the optimal MDP strategy, which might not be reachable for a given synthesis problem with bounded depth. We split the results according to the size of the MDP and the depth $k$ to highlight the scalability advantage of \dtpaynt over \omdt. Detailed statistics for the experiments are reported in Appendix~\ref{app:results}.

\paragraph{On the smaller models, both tools perform very similarly.} With growing $k$, both tools produce DTs with an increasing performance that gets close to the optimal value. \omdt finds slightly better DTs on two models: \emph{lake-12} and \emph{sys-2}.

\paragraph{On the majority of the larger models, \dtpaynt outperforms \omdt significantly.} In fact, there is no larger model in our data set where \omdt performs better at any depth $k$. \omdt is only able to keep up with \dtpaynt on the model \emph{inventory} and for smaller $k$ on \emph{firewire-3}. On the other models, the performance gap is significant. The greatest difference was observed on models \emph{resource}, \emph{pnueli} and \emph{csma} where \dtpaynt is able to find a tree with value within 1\% of the optimum while \omdt struggles to improve on the best 0-DT. Another observation is that with increasing $k$, \dtpaynt mostly improves, but \omdt struggles even more to find good DTs.
This can be seen by the higher occurrence of the green dots 
at the bottom of the plot. These results clearly demonstrate the scalability advantage of \dtpaynt compared to \omdt.

\paragraph{Runtimes.} In most cases, both tools reach the 20-minute timeout: \dtpaynt is not able to completely explore all DTs up to the given depth; \omdt is not able to reduce the gap between the lower and upper bound in the underlying MILP below the given precision. However, we observe (see Fig.~\ref{fig:appendix-omdt-runtimes} in Appendix~\ref{app:runtimes}) that typically \dtpaynt finds better DTs significantly faster.

\paragraph{Conclusion.} 
For simple synthesis problems (smaller MDPs or depths), our approach is competitive with the monolithic MILP formulation. For larger models, solving MILP is no longer tractable while \dtpaynt is able to find high-quality DTs and keep its performance even for higher values of $k$. This clearly demonstrates the significant
advantage of \dtpaynt over \omdt.

\vspace{-.5em}
\subsection*{Q2: Size and Value Trade-off Comparison with \dtcontrol}

\dtcontrol, in contrast to \omdt and \dtpaynt, solves the policy mapping problem; it takes a policy \sched and constructs a DT representing~\sched. It favors scalability over minimality, i.e., it does not search for the smallest DT representing~\sched. Thus, \dtcontrol excels in finding a DT representing the given optimal policy even for large MDPs. In this section, we demonstrate the key advantage of \dtpaynt: it is able to effectively control the trade-offs between the size and value of the resulting DTs at the cost of scalability.

\paragraph{Setting.} Fig.~\ref{fig:dtcontrol-bar-plot} reports the trade-offs achieved by DTs of different depths compared to the DT produced by \dtcontrol. The figure contains results for 13 models (out of the 21 models) that demonstrate the key observations we made; complete results can be found in Appendix~\ref{app:results}. As before, the timeout for each experiment (model and depth) is 20 minutes.
The upper part compares the normalized values achieved by the particular DTs produced by \dtpaynt. 
The lower part of Fig.~\ref{fig:dtcontrol-bar-plot} compares the sizes of DTs (the number of inner nodes). Recall that \dtpaynt optimizes the depth, and thus, the number of inner nodes can be larger than in \dtcontrol, although the depth of the DT is smaller.

\paragraph{Preprocessing for \dtcontrol}
In order to provide a fair comparison, we perform the following preprocessing of the policy \sched for \dtcontrol: we remove from its tabular representation 
the states in the goal set $G$ and the states that are unreachable in \imc. This helps \dtcontrol to find significantly (5-10 times in some cases) smaller DTs. \dtpaynt performs this simplification implicitly.% (see Sec.~\ref{sec:smt}).

\setlength{\textfloatsep}{15pt}
\begin{figure}[t]
    \centering
    \resizebox{\textwidth}{!}{\input{figures/plot-q2-bar-plot.pgf}}
    \vspace{-2em}
    \caption{The plot shows the trade-off between the value and size of the DTs found by \dtpaynt and \dtcontrol. The left part contains smaller models, and the right part contains larger ones. The upper part shows the normalized values of the synthesized DTs (the same normalization as in Q1). \dtcontrol maps an optimal policy with the value 1 (not shown).
    The lower part shows the number of inner nodes using a logarithmic scale. We report the performance of \dtpaynt for depths 1 to 8 (if a DT with a normalized value above 0.95 is obtained, we exclude the subsequent depths on the given model to simplify the plot). Numbers in brackets above the \dtcontrol bars denote the depth of the DT.}
    \label{fig:dtcontrol-bar-plot}
\end{figure}

\paragraph{Results.} On the majority of the models, \dtpaynt finds DTs that provides good trade-offs between the size and value.
On the models \emph{inventory, csma} and \emph{rabin}, \dtpaynt finds a DT that achieves the optimal value, while the number of inner nodes is reduced by 2 times (for \emph{rabin}) up to 16 times (for \emph{inventory}) compared to the DTs constructed by \dtcontrol. On the other 6 models, \dtpaynt is able to find a DT with a normalized value better than 0.9 while producing a more compact DT compared to \dtcontrol. These DTs are, on average, 3-5 times smaller. The most interesting trade-off can be observed on the model \emph{consensus}, where the DT produced by \dtpaynt is 38-times smaller (note that the bottom plot in Fig.~\ref{fig:dtcontrol-bar-plot} uses logarithmic scale) while its normalized value is above 0.93.  On the models \emph{sys-tree} and \emph{tictactoe}, \dtpaynt finds a DT that achieves the optimal value, but it has more inner nodes compared to the one from \dtcontrol. 
For the models where good DTs require depth $k\geq 8$, such as \emph{lake-12}, \emph{sys-2} and \emph{maze-7}, \dtpaynt is not capable of finding a DT that achieves a value close to 1, but it still provides a reasonable trade-off. In the case of \emph{maze-7}, \dtpaynt finds a 25 times smaller DT that achieves normalized value~0.76.

\paragraph{Runtimes.} The runtimes for \dtcontrol are less than 2 seconds for all benchmarks. While \dtpaynt usually finds the best DT long before the 20-minute timeout is reached, it would still benefit from a longer timeout for larger depths.

\paragraph{Conclusions.} We show the advantages and limitations of using \dtpaynt in comparison to \dtcontrol. While \dtpaynt iteratively performs analysis of the underlying MDP and solves complicated SMT queries that hinder the scalability, it is able, in contrast to \dtcontrol, to effectively explore alternative policies. This is essential for constructing smaller and thus more explainable DTs.

\subsection*{Q3: Scalability of \dtpaynt on very large MDPs}
The goal of the experiment is to demonstrate that \dtpaynt can handle much larger MDPs than those considered in Q1 and Q2 provided that a small DT with the desired performance exists.
There are generally three reasons why small DTs suffice: (1)~The optimal policy induces only a very small MC, and thus the optimal DT needs to encode the decision only for a small subset of states. (2)~In many states, there is only one available action, and thus the optimal DT does not need to encode decisions in these states. 
(3)~In many states, playing the random action $\actrandom$ that uniformly selects one of the available actions is sufficient to obtain the optimal performance. $\dtpaynt$, in contrast to \omdt and \dtcontrol, naturally leverages all three features; recall Def.~\ref{def:induced-policy} (induced policy) and our SMT encoding.
Note that the experimental setting in Q1 and Q2 was designed with the goal of providing a fair comparison among the tools, and this advantage of \dtpaynt was not fully exploited.

\paragraph{Setting.} We consider MDPs with up to 1M states. We take MDPs from the QComp benchmark set and scale them up using their parameters. We run \dtpaynt for depths 0 to 4. Table~\ref{tab:q3-big-mdps} reports selected results that demonstrate the key observations within Q3. We compare the results with \dtcontrol that maps the optimal policy $\schedopt$. To provide a fair comparison, we exclude from $\schedopt$ the unreachable states (as in Q2) and states where only one action is available (such states were not in Q2). We report the size of the pre-processed policy as $|\sigma^*_{rel}|$\footnote{$|\sigma^*_{rel}|$ can be further reduced, at the expense of precision, by considering only states with a high probability of being visited~\cite{brazdil2015counterexample}. We left this idea for future work.}.

\begin{table}[t]
\setlength{\tabcolsep}{2pt}
\centering
\small
\begin{tabular}{l@{\hskip 4pt} cc@{\hskip 16pt} rrr@{\hskip 16pt} rrr}
\toprule
\multirow{2}{*}{model} & \multirow{2}{*}{$|S|$} & \multirow{2}{*}{choices} & \multicolumn{3}{c@{\hskip 16pt}}{\dtpaynt} & \multicolumn{3}{c}{\dtcontrol} \\
\cmidrule(lr{2em}){4-6}\cmidrule(lr){7-9}
 & & & value & nodes & time & $|\sigma^{*}_{rel}|$ & nodes & time \\
\midrule
ij-20 & 1M & 10M & 1 & 0 & 547s & 624k & 393k & 210s \\
pnueli-zuck-5 & 308k & 1.7M & 1 & 0 & 103s & 2395 & 1258 & 1s \\
firewire-f-36 & 212k & 479k & 1 & 14 & 531s & 376 & 12 & 1s \\
pacman-30 & 853k & 1.1M & 0.76 & 6 & 1360s & 673 & 144 & 1s \\
firewire-t-3-600 & 1.1M & 1.5M & 0.85 & 8 & 3135s & 1147 & 12 & 1s \\

\bottomrule
\end{tabular}
\vspace{.5em}
\caption{This table presents selected results of the synthesis for large MDPs. 
It shows the normalized values (as before), the number of inner nodes and the time it took to find the DTs. $|\sigma^*_{rel}|$ is the size of the preprocessed policy \schedopt.}
\vspace{-1em}
\label{tab:q3-big-mdps}
\end{table}

\paragraph{Results.} 
Except for the \emph{ij-20} model, the pre-processing drastically reduces the size of $\schedopt$. The optimal DTs thus need to encode significantly less state compared to the size of $S$. For the first two models, \dtpaynt finds a 0-DT with the optimal value. It selects a single action $\alpha \neq \actrandom$ and exploits the fact that playing $\actrandom$ in the states where $\alpha \notin \Act(s)$ is optimal. Although trivial policies suffice in these two models, they showcase the ability of \dtpaynt to steer towards simpler policies, while \dtcontrol returns very complicated DTs despite the pre-processing. Note that the QComp benchmark contains several other synthesis problems where simple randomization is sufficient. In the other models in the table, the randomization is not optimal.
For the \emph{firewire-f-36} model, \dtpaynt finds an optimal DT with a depth lower than \dtcontrol (4 vs. 7), but it has two more inner nodes.
For the last two models, \dtpaynt finds DTs that achieve a reasonable performance while being smaller (24-fold for \emph{pacman-30}) compared to the DTs produced by \dtcontrol.

\paragraph{Conclusions.} \dtpaynt is able to find small DTs with a good performance even for very large MDPs with hundreds of thousands of states. In some cases, it finds optimal DTs that are smaller than the DTs constructed by \dtcontrol.

\subsection*{Q4: Using \dtpaynt as a reduction procedure for large DTs}

Finally, we discuss an alternative use of \dtpaynt that demonstrates its broader applicability even to MDPs that require large DTs. 
The main idea is that \dtpaynt can be used as a procedure for minimizing decision trees by analyzing individual sub-trees in a DT and trying to find an adequate smaller alternative for a given sub-tree. 
We give a short and simple outline below. For a more detailed discussion on this use of \dtpaynt, refer to our follow-up paper~\cite{andriushchenko25dtnest}.

\paragraph{Using \dtpaynt in a compositional manner.}
We take a large DT $\tree$ representing 
an optimal policy for the given MDP and investigate whether \dtpaynt can find a smaller variant of this DT. Therefore, we consider $\tree$ and its sub-tree $\tree_n$ given by an inner node $n$ in $\tree$ (i.e. $n$ is the root of $\tree_n$). We build a tree template $\template_n$ by fixing all nodes outside $\tree_n$ to coincide with $\tree$.
For the nodes in $\tree_n$, we define $\Gamma$ and $\Delta$ as in Def.~\ref{def:template}.
The size of the sub-tree $\tree_n$ indeed determines the size of $\template_n$, and also the complexity of the synthesis problem: $\template_n$ defines a partial policy and thus induces a (possibly much smaller) subset of states that has to be considered in the synthesis process. \dtpaynt can straightforwardly take $\template_n$ and synthesize a new smaller tree $\tree_n'$ to replace $\tree_n$ in $\tree$ to reduce the total number of inner nodes. We can allow a certain error with respect to the optimal policy to increase the possible reduction. This process can be applied iteratively on different sub-trees on the original DT $\tree$.

\paragraph{Setting.} We consider two more complex variants of the models from Q1 and Q2: i) \emph{consensus-Q4} has only 23k states, but the DT found by \dtcontrol has depth 19 and 841 inner nodes; ii) \emph{csma-Q4} has 1.5M states and the DT found by \dtcontrol has depth 19 and 236 nodes. In both cases, \dtpaynt fails to find a small DT with a reasonable performance. To obtain a smaller DT with almost optimal value, we take the DT $\tree$ produced by \dtcontrol and iteratively run \dtpaynt on templates induced by all sub-trees of $\tree$ of depth 7 (and consequently on depth 6 if all sub-trees of depth 7 were processed).  We allow an \emph{absolute error} of 1\% with respect to the optimal policy.

\paragraph{Results.}  
For the \emph{consensus-Q4} model, \dtpaynt analyzed in 1 hour 54 different sub-trees and found 37 smaller sub-trees that were used to replace the original ones in $\tree$. The resulting DT has a depth of 16 and 452 inner nodes, signifying a 46\% decrease in size, while its normalized value is 0.99. For the \emph{csma-Q4} model, \dtpaynt analyzed in 30 minutes 34 different sub-trees and found 24 smaller sub-trees. The resulting DT has depth 6 and 22 inner nodes, signifying a 90\% decrease in size, while its normalized value is 0.99.

\paragraph{Conclusions.} 
Although the scalability of \dtpaynt itself is limited if a deep tree is needed to achieve the desired quality of the resulting policy, we demonstrate that \dtpaynt can be effectively used to reduce the size of large DTs even in large MDPs. Thus \dtpaynt provides an important contribution towards the automated construction of near-optimal smaller DTs for complex problems.

\section{Related Work}
The related work in learning~\cite{ashok2021dtcontrol} and MILP-based synthesis~\cite{vos2023optimal} of decision trees as well as in deductive controller synthesis~\cite{andriushchenko2022inductive} is discussed in the introduction.

\paragraph{Learning concise representation of neural policies.} With the boom of explainable reinforcement learning, various methods for learning concise representations of neural policies have been proposed. Imitation learning methods such as VIPER~\cite{bastani2018verifiable} extract a DT from a more complex teacher policy using a supervised learning paradigm. As shown in~\cite{vos2023optimal}, imitating a complex policy using a small DT can lead to poor performance as the limited capacity of the DT is used ineffectively. A different approach for overcoming this limitation has been recently proposed in~\cite{topin2021iterative}. The authors introduce a new type of MDP (so-called iterative bounding MDP) where each policy corresponds to a DT policy for the base MDP. Especially for small DTs, this approach significantly outperforms VIPER. A different line of work focuses on learning a programmatic representation of policies using an oracle in the form of neural policy. In the seminal paper~\cite{DBLP:conf/icml/VermaMSKC18}, the authors showed how to search over programmatic policies that minimize the distance from the oracle. More recently, a fast distillation method that uses regularized oblique trees to produce tree programs that fits neural oracle has been proposed~\cite{kohler2024interpretable}.

\paragraph{Beyond DTs.} Recently, alternative representations of policies in MDPs have been studied. In~\cite{batz2024programmatic}, the authors establish a tight connection between program-level construction of strategies resolving nondeterminism in probabilistic programs and finding good policies in (countably infinite) MDPs. This enables a direct construct of programmatic policies. A different line of work introduces templates of decision diagrams using hierarchical control structures with under-specified entities
to encapsulate and reuse common decision-making patterns~\cite{dubslaff2024template}. In contrast to our templates used to effectively reason about sets of DTs, the hierarchical decision diagrams aim at a more concise and explainable representation of the policies.

\section{Conclusion}
We present \dtpaynt, a novel approach to synthesize small DTs providing good trade-offs between quality and size. 
Our experiments demonstrate clear advantages with respect to the state-of-the-art. 
In the future, we will investigate how to improve the scalability: (1) exploit counterexamples~\cite{cegis-journal} in the synthesis of DTs, (2)~symbiotically combine \dtpaynt with \dtcontrol in a more efficient synthesis loop to improve over the ideas in Q4, and (3)~construct the DTs only for a subset of most relevant states as in~\cite{brazdil2015counterexample}.

\subsubsection{Acknowledgments.}
\inlinegraphics{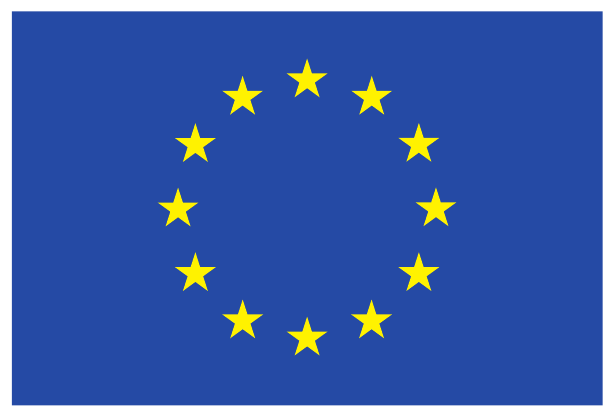} This work has been executed under the project VASSAL: ``Verification and Analysis for Safety and Security of Applications in Life'' funded by the European Union under Horizon Europe WIDERA Coordination and Support Action/Grant Agreement No. 101160022  and the NWO VENI Grant ProMiSe (222.147). Additionally, this work has been funded by the Czech Science Foundation grant \mbox{GA23-06963S} (VESCAA) and the IGA VUT project FIT-S-23-8151.

\subsubsection{Disclosure of Interests.}
The authors have no competing interests to declare that are relevant to the content of this article.

\newpage

\bibliographystyle{splncs04}
\bibliography{bibliography}

\newpage

\appendix

\section{Proofs}
\label{appendix:proofs}

In the following, assume an MDP $M = \mdpT$ and a tree template $\template = \templateT$.

\begin{lemma}
\label{lemma:family}
Assume a parameterization $f$ and a set $\family$ of parameterizations. Then $f \models \fdom(\family)$ iff $f \in \family$.
\end{lemma}
\begin{proof}
Follows directly from Def.~\ref{def:parameterization}.
\end{proof}

\begin{corollary}
\label{corollary:family-primed}
Assume parameterizations $f,f'$ and two sets $\family,\family'$ of parameterizations where $\family'$ is a set identical to $\family$ (but associated with primed variables). Then, $f,f' \models \fdom(\family) \land \fdom(\family') \Leftrightarrow f \in \family \land f' \in \family' \Leftrightarrow f,f' \in \family$.
\end{corollary}

\begin{lemma}
\label{lemma:instantiation}
Let $f = (\decisionfun, \boundfun, \actionfun)$, where each function $\mathcal{Z} \in \{\decisionfun, \boundfun, \actionfun\}$ is expressed using parameters $\mathcal{Z}_n$ to denote $\mathcal{Z}(n)$.
Let $s \in S$ be an arbitrary state and let $n = \leaf[\template(f)](s)$. Then $\sched[\template(f)](s) = f(\actionfun_n)$ if $f(\actionfun_n) \in Act(s)$ and $\sched[\template(f)](s) = \actrandom$ otherwise.
\end{lemma}
\begin{proof}
Recall from Def.~\ref{def:parameterization} that $\template(f) = (T, \gamma_{\decisionfun,\boundfun},\delta_\actionfun)$ where
$\delta_\actionfun(n) \coloneqq f(\actionfun_n)$.
From Def.~\ref{def:induced-policy} it follows that $\sched[\template(f)](s) = \delta_\actionfun(n) = f(\actionfun_n)$ if $f(\actionfun_n) \in \Act(s)$ and $\sched[\template(f)](s) = \actrandom$ otherwise.
\end{proof}

\subsection{Proof of Corollary~\ref{corollary:random}}
\label{appendix:proof:random}

In the following, we assume that the optimizing direction of the specification is to maximize the value of the policy but it is clear that this proof works even for minizing specifications. Let $M$ be an MDP with the optimal value $V(M) = V(\sigma*)$ achieved by a memoryless deterministic scheduler $\sigma^{*}$ (as per Puterman~\cite{Put94} this holds for all the specifications considered in this paper). Let us consider an MDP $M'$, which was created from $M$ by adding arbitrarily many new actions, which in each state behave as some distribution over available actions in that state. We prove that the optimal values are equal i.e. $V(M') = V(M)$ by proving:

\begin{enumerate}
    \item $V(M') \geq V(M)$ - this follows directly from the construction of $M'$ since its action space is a superset of $M$.
    \vspace{1em}
    \item $V(M') \leq V(M)$ - if there exists a scheduler $\sigma'$ in $M'$ such that $V(\sigma') > V(M)$, then $\sigma'$ can be transformed into a randomized scheduler $\sigma_{rand}$ for the original MDP $M$ with the same value. This transformation simply changes all the decisions of the scheduler that chose the newly added actions in $M'$ and replaces them with a distribution over actions in $M$ that corresponds to these actions. This means that $V(\sigma_{rand}) > V(\sigma*)$ in $M$, which contradicts~\cite{Put94} Theorem 7.9.1 that the optimal value is achievable by a memoryless deterministic scheduler.
\end{enumerate}

\subsection{Proof of Theorem~\ref{thm:smt}}
\label{appendix:proof:encoding}

% Assume a parameterization $f$ and a policy $\sched \in \schedulers$.
% Then $f \models \fpol(\family,\sched)$ iff $f \in \family$ and $\forall s \in S \colon \left(\sched(s) \neq \bot \right) \rightarrow \left(\sched(s) = \sched[\template(f)] \right)$.

Assume a policy $\sched \in \schedulers$, a set $\family \subseteq 
\superfamily$ of parameterizations and an arbitrary $f \in \superfamily$.
By Lemma~\ref{lemma:family}, to prove equivalence in Thm.~\ref{thm:smt}, it is sufficient to assume that $f \in \family$ and show that  $f \models \bigwedge_{s \in S, \sched(s) \neq \bot} \fact_{s,\sched(s)}$ iff $\sched[\template(f)] = \sched$.
In particular, we will show that for an arbitrary state $s \in S$, $\sched(s) \neq \bot$, it holds that $f \models \fact_{s,\sched(s)}$ iff $\sched(s) = \sched[\template(f)](s)$.

\paragraph{($\Rightarrow$)}
Assume $f \models \fact_{s,\sched(s)}$, i.e. $f \models \bigwedge_{n \in L} \left(\fsel_{s,n} \rightarrow \fact_{s,\sched(s),n} \right)$.
Observe the left-hand side $\fsel_{s,n}$ of the implication. Substituting each variable $x \in \params$ with its value $f(x)$, we obtain

$$\fsel_{s,n}(f) = \bigwedge\textstyle_{i=0}^{k-1} \bigwedge\textstyle_{j=1}^{|\variables|} \quad \big(s(v_{f(\decisionfun_{n_i})})\meetsin f(\boundfun_{n_i}) \big)$$
Note that $\meetsin$ can be either $\leq$ or $>$ depending on the choice of the leaf node $n \in L$.
Thus, there exists exactly one such $n$ for which expression $\fsel_{s,n}(f)$ is true; in fact, $n = \leaf[\template(f)](s)$.
We obtain $f \models \fsel_{s,n}$, from which $f \models \fact_{s,\sched(s),n}$.

\begin{itemize}
\item Case $\sched(s) \neq \actrandom$. From $f \models \fact_{s,\sched(s),n}$ we obtain $f(\actionfun_n) = \sched(s)$ and thus, by Lemma~\ref{lemma:instantiation}, we get $\sched[\template(f)](s) = f(\actionfun_n) = \sched(s)$.
\item Case $\sched(s) = \actrandom$. From $f \models \fact_{s,\sched(s),n}$ we obtain $f(\actionfun_n) \in \{\actrandom\} \cup \Act {\setminus} \Act(s)$, i.e. either $f(\actionfun_n) = \actrandom$ or $f(\actionfun_n) \not \in \Act(s)$; in either case, by Lemma~\ref{lemma:instantiation}, we get $\sched[\template(f)](s) = \actrandom = \sched(s)$.

\end{itemize}

\paragraph{($\Leftarrow$)} The proof is analogous to ($\Rightarrow$).

% \begin{align*} \fsel_{s,n} \coloneqq\quad  \textstyle
%  &\bigwedge\textstyle_{i=0}^{k-1} \bigwedge\textstyle_{j=1}^{|\variables|} \quad \big(j 
%  = \decisionfun_{n_i}\big) \rightarrow \big(s(v_j)\meetsin \boundfun_{n_i}\big). 
% \end{align*}

% \begin{align*}
% \fact_{s,\act} & \coloneqq \quad  \bigwedge_{n \in L}  \fsel_{s,n} \rightarrow \fact_{s,\act,n}, \qquad \text{using} \\
% \fact_{s,\act,n} & \coloneqq
% \begin{dcases*}
% \actionfun_{n} = \act & if $\act \neq \actrandom$ \\
% \actionfun_{n} \in \{\actrandom\} \cup \Act {\setminus} \Act(s)  & otherwise \\
% \end{dcases*}    
% \end{align*}

\subsection{Proof of Theorem~\ref{thm:harmonization}}
\label{appendix:proof:harmonization}

Assume parameterizations $f_1,f_2 \in \superfamily$, two (identical) sets of parameterizations $\family,\family'$, a policy $\sched \in \schedulers$, and an evaluation $f_h \in \integers$ of the variable $h$.
Then $f_h,f_1,f_2 \models \fpol^H(\family,\sched)$ iff $f_1,f_2$ are harmonizing for $\sched$ via a harmonizing variable $x_{f_h} \in \params$.

Assume two wo sets $\family,\family' \subseteq \superfamily$ of parameterizations where $\family'$ is a set identical to $\family$ (but associated with primed variables), a policy $\sched$ that is not $\template(\family)$-impelementable, two parameterizations $f_1,f_2 \in \superfamily$ and an evaluation $f_h \in \integers$ of the variable $h$.

By Def.~\ref{def:harmonization} and Cor.~\ref{corollary:family-primed}, to prove equivalence in Thm.~\ref{thm:harmonization}, it is sufficient to assume that $f_1,f_2 \in \family$ and show that  $f_h,f_1,f_2 \models \fharm(\family) \land \bigwedge_{s \in \criticalstates} \left( \fact_{s, \sched(s)} \lor \fact'_{s,\sched(s)} \right)$ iff $\exists!\, x \in \params : f_1(x) \neq f_2(x)$ and for any state $s \in \criticalstates$, $\sched(s) \in \{\sched[\template(f_1)](s),\sched[\template(f_2)](s)\}$.
We will show this by proving two equivalences: (1) $f_h,f_1,f_2 \models \fharm(\family)$ iff $\exists!\, x \in \params : f_1(x) \neq f_2(x)$, and (2) given arbitrary state $s \in \criticalstates$, $f_1,f_2 \models \left( \fact_{s, \sched(s)} \lor \fact'_{s,\sched(s)} \right)$ iff $\sched(s) \in \{\sched[\template(f_1)](s),\sched[\template(f_2)](s)\}$.

\vspace{0.4em}\noindent(1, $\Rightarrow$) $f_h,f_1,f_2 \models \fharm(\family) \Leftrightarrow \bigwedge_{1 \leq i \leq |\params|} f_h \neq i \,\rightarrow\, f_1(x_i) = f_2(x'_i) \Leftrightarrow f_1$ and $f_2$ are equal for all variables except $x_{f_h}/x'_{f_h}$. Since $f_1 \neq f_2$ -- otherwise it would hold that $f_1 \models \bigwedge_{s \in \criticalstates, \sched(s) \neq \bot} \fact_{s,\sched(s)}$, which contradicts the assumption that $\sched$ is not $\template(\family)$-impelementable (on \criticalstates) -- then it must hold that $f_1(x_{f_h}) \neq f_2(x'_{f_h})$, i.e. (in terms of non-primed variables), $f_1(x_{f_h}) \neq f_2(x_{f_h})$.

\vspace{0.4em}\noindent(1, $\Leftarrow$)  Assume there exists a unique $x_{f_h} \in \params$ s.t. $f_1(x_{f_h}) \neq f_2(x_{f_h})$, i.e. (in terms of primed variables) $f_1(x_{f_h}) \neq f_2(x'_{f_h})$. Then $\bigwedge_{1 \leq i \leq |\params|} f_h \neq i \,\rightarrow\, f_1(x_i) = f_2(x'_i)$, i.e. $f_h,f_1,f_2 \models \fharm(\family)$.

\vspace{0.4em}\noindent(2, $\Leftrightarrow$) Assume $f_1,f_2 \models \left( \fact_{s, \sched(s)} \lor \fact'_{s,\sched(s)} \right)$, i.e. $f_1 \models \fact_{s, \sched(s)}$ or $f_2 \models \fact'_{s,\sched(s)}$. According to Thm.~\ref{thm:smt}, this is equivalent to $\sched(s) = \sched[\template(f_1)](s) \lor \sched(s) = \sched[\template(f_2)](s)$, i.e.~$\sched(s) \in \{\sched[\template(f_1)](s),\sched[\template(f_2)](s)\}$.

\subsection{Proof of Theorem~\ref{thm:np-hardness}}
\label{appendix:thm:np-hardness}

The main idea of the proof is the reduction from the Exact Cover by 3-sets (X3C) to the question of whether there is a decision tree of depth $k$ for MDP $M$ that implements a policy which reaches a goal state with probability above $0.5$. The X3C problem is known to be NP-complete. 
X3C asks whether for a set $U$ with $|U| = 3q$ and a collection $T = \{T_1,T_2,...,T_j\}$ with $|T_i| = 3$, $T_i \subseteq U$, we can find a subset $T'$ of $T$ such that every element $u \in U$ occurs in exactly one member of $T'$ (i.e., $T'$, is an exact cover of $U$).

\paragraph{Reduction.} In the reduction, we assume $U = \{1,2,3,...,n\}$ is an ordered set. For every instance of X3C we can construct an MDP $M = (S, s_{0}, Act, P)$ where $S = U \cup \{\bot, n+1, n+2, n+3, n+4\}, s_{0} = 1, Act = U \cup \{n+1, n+2, n+3\}$, and the transition function $P$ defined s.t.:
$\forall s \in (U \cup \{n+1, n+2, n+3\})$ and $\forall a \in (U \cup \{n+1, n+2, n+3\}): P(s,a,s+1) = 1$ if $s = a$ and $P(s,a,\bot) = 1$ otherwise.
The state $n+4$ is the goal state, and the state $\bot$ is a sink state. Fig.~\ref{fig:x3c:mdp} showcases the main idea of the MDP construction.

\begin{figure}[t]
    \begin{subfigure}{.38\textwidth}
    \centering
    \includegraphics[width=1\textwidth]{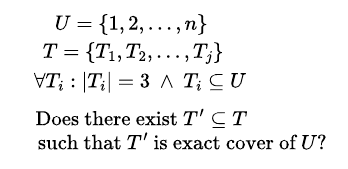}
    \caption{}
    \label{fig:x3c:x3c}
    \end{subfigure}%
    \begin{subfigure}[b]{.6\textwidth}
    \centering
    \includegraphics[width=0.95\textwidth]{figures/x3c-to-mdp.pdf}
    \caption{}
    \label{fig:x3c:mdp}
    \end{subfigure}%
    
    \begin{subfigure}[b]{.98\textwidth}
    \centering
    \includegraphics[width=1\textwidth]{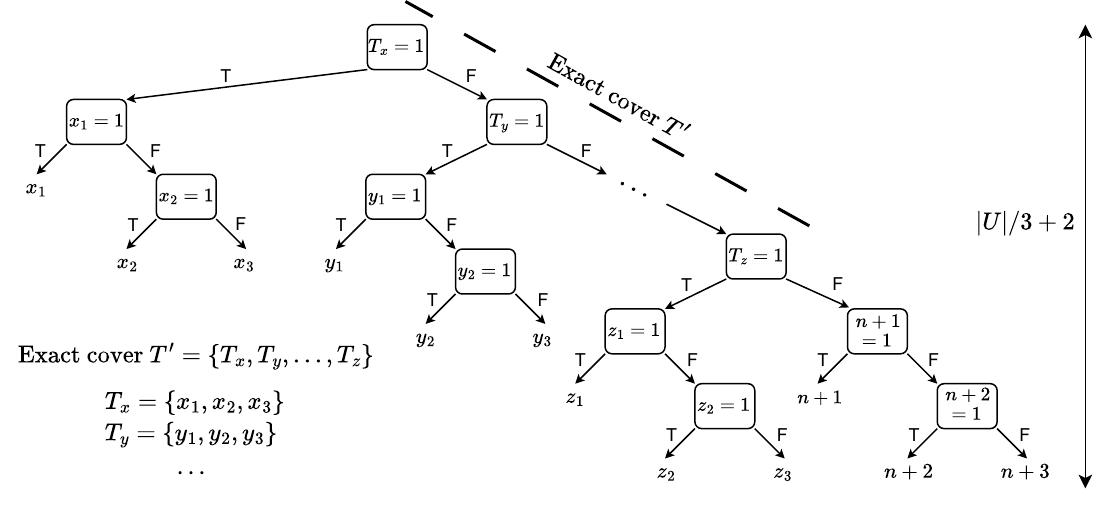}
    \caption{}
    \label{fig:x3c:tree}
    \end{subfigure}%
    \vspace{-0.5em}
\caption{
(a) The Exact Cover by 3-sets (X3C) problem formulation. (b) MDP construction for the reduction from X3C problem formulation. Note that the values of state variables $T_i$ were chosen arbitrarily and serve only as an example. This figure is the same as Fig.~\ref{fig:x3c-to-mdp}. (c) DT corresponding to some exact cover $T'$ with a depth $|U|/3+2$. 
}
\label{fig:x3c}
\end{figure}

The main idea is that in each state $s$ (except $\bot$ and $n+4$), a unique action for that state needs to be selected to reach the next state $s+1$ and to avoid reaching the sink state (reaching the sink state means the value of the policy is 0) i.e. the decision tree that achieves value above $0.5$ indeed needs to distinguish all these n+3 states into separate leaf nodes.

\paragraph{Definition of state variables.} The MDP will have state variables set $V = T \cup S$, and the state mapping $s: V \rightarrow Z$ such that $\forall T_i \in T: s(T_i) = 1$ if $s \in T_i$ and $s(T_i) = 0$ otherwise, 
and $\forall v \in S: s(v) = 1$ if $s = v$ and $s(v) = 0$ otherwise. 
Thus, for each state $s$, only those predicates $T_i$ will be true where the element of $U$ corresponding to state $s$ is in $T_i$, and there will be a set of predicates for each state that are only true in one unique state.

\paragraph{Definition of threshold.} Finally, the bound on the depth $k$ will be set as $|U|/3 + 2$. Notice that choosing the predicate corresponding to some subset $T_i$ is able to distinguish at most three states from the rest in an inner node of the tree.

\paragraph{Correctness.} As shown in [28], if the size $|U| > 7$, then there must be at least $k-2$ (or $|U|/3$) decision predicates corresponding to sets in $T$ chosen in the DT. 
If at least two of the predicates chosen in these first $k-2$ decisions are true for the same state, then the inner node at depth $k-2$, which you get to when you follow the "always false" path from the root node, will need to decide about at least four states that are not yet distinguished in the upper levels. 
Three of the states are $n+1$, $n+2$, and $n+3$ since no predicates corresponding to $T$ hold true in them and one state for which no predicate held true in the upper levels. 
It is not possible to distinguish between these four states using a 2-DT (or equivalently, using the remaining two levels of the k-DT, as not to break the $k$-bounded depth) since one can only use predicates which hold true in one of the states and are false in the rest. 
This means the only way a $k$-depth tree which distinguishes all states into unique leaf nodes exists is if there was an exact cover in T, and this exact cover can be retrieved from the tree by looking at the predicates chosen in the first k-2 levels of the tree. The $k$-DT showcasing this idea is shown in Fig.~\ref{fig:x3c:tree}.

\vspace{0.5em}
\noindent
Note that in the reduction, we can introduce further template restrictions other than the depth of the tree e.g. for the tree to only use those predicates that correspond to sets in $T$ in some of the inner nodes.

\subsection{Proof of Theorem~\ref{thm:sound}}
\label{appendix:thm:sound:proof}

We reiterate that any nontrivial splitting makes Algorithm~\ref{alg:refinement} terminate since the family \superfamily is finite.
Let $\family_{\max} \coloneqq \mathrm{arg\,max}_{f \in \superfamily}V(f)$ be a set of optimal assignments, i.e. correct outputs of Algorithm~\ref{alg:refinement}.
Note that $\family_{\max} \neq \emptyset$ and thus Algorithm~\ref{alg:refinement} is complete iff it is sound.
To show the soundness, it is sufficient to inspect ll.\ref{alg:refinement:prune}-\ref{alg:refinement:sat}, where a subset can be pruned, and show that on these lines we do not discard some $f^* \in \family_{\max}$.
Assume an arbitrary iteration of the algorithm where subset $\family \subseteq \superfamily$ is analyzed.
Let \sched be a maximizing policy for $M(\family)$ obtained on~l.\ref{alg:refinement:solve}.
From Proposition.~\ref{proposition:family-mdp} it follows that $V(\sched) \geq \max_{f \in \family}V(f)$.
If~\family is pruned on~l.\ref{alg:refinement:prune} because $V(\sched) \leq V_{best}$, then $\max_{f \in \family}V(f) \leq V_{best}$, i.e.~no $f \in \family$ has better value than $f_{best}$.
Otherwise, if~\family is pruned on~l.\ref{alg:refinement:sat}, $f_{best}$ is updated with assignment $f$ for which $V(f) = V(\sched) \geq max_{f \in \family}V(f)$.
In both cases, no $f^* \in \family_{\max}$ is discarded unless $f_{best} \in \family_{\max}$.
%\ra{I don't like this proof but someone should read this.}

\section{Benchmarks}
\label{appendix:model-info}

This section contains information about all of the 21 considered models. This information is summarized in Table~\ref{tab:appendix-model-info}. All the benchmarks are available at \url{https://anonymous.4open.science/r/dt-synthesis-cav-25-5B31}.

\begin{table}[h]
\setlength{\tabcolsep}{2pt}
\centering
\begin{tabular}{l rrrr@{\hskip 32pt}l rrrr}
\toprule
\multirow{1}{*}{model} & vars & $|S|$ & $|Act|$ & choices & \multirow{1}{*}{model} & vars & $|S|$ & $|Act|$ & choices \\
\midrule
maze-7 & 9 & 2039 & 5 & 10195 & maze-steps & 5 & 198 & 5 & 990 \\
3d & 3 & 125 & 7 & 875 & blackjack & 3 & 533 & 3 & 1599 \\
lake-4 & 2 & 16 & 5 & 80 & lake-8 & 2 & 64 & 5 & 320 \\
lake-12 & 2 & 144 & 5 & 720 & inventory & 1 & 101 & 101 & 10201 \\
sys-1 & 8 & 256 & 10 & 2560 & sys-2 & 8 & 256 & 10 & 2560 \\
sys-tree & 7 & 128 & 9 & 1152 & tictactoe & 27 & 2424 & 10 & 24240 \\
traffic & 4 & 361 & 3 & 1083 & consensus & 5 & 6161 & 14 & 86254 \\
csma & 11 & 7959 & 17 & 135303 & firewire & 10 & 4094 & 14 & 57316 \\
philos & 4 & 3193 & 53 & 169229 & pnueli & 3 & 1950 & 63 & 122850 \\
rabin & 19 & 10241 & 17 & 174097 & resource & 7 & 3292 & 5 & 16460 \\
wlan & 11 & 3127 & 34 & 106318 & consensus-Q4* & 7 & 22656 & 26 & 60544  \\
csma-Q4* & 15 & 1.5M & 25 & 1.6M & & & & & \\
\bottomrule
\end{tabular}
\vspace{1em}
\caption{Information about the models used in the experimental evaluation. *~denotes the models which were not part of the Q1 and Q2 experiments.}
\label{tab:appendix-model-info}
\end{table}

\section{Complete experimental results}
\label{app:results}

% This section provides the experimental results for all of the models considered in our benchmark set.

% \subsection{Comparison with \dtcontrol}
% \label{appendix:q2}

% Table~\ref{tab:appendix-q2} contains the results of all 22 benchmarks we considered for the policy mapping problem.

% \begin{table}[h]
% \setlength{\tabcolsep}{2pt}
% \centering
% \input{tables/appendix-q2}
% \vspace{1em}
% \caption{Complete comparison of $\dtmap$ with $\dtcontrol$ on the full benchmark set. The upper part includes models with less than 3k choices, and the bottom part includes models with more than 3k choices. The depth ${>k}$ 
% indicates that \dtmap proved that there is no $k$-DT representing the given policy.
% Bold numbers indicate the DT with the smallest number of inner nodes found for the given policy.}
% \label{tab:appendix-q2}
% \end{table}

% \subsection{Comparison with \omdt}
% \label{appendix:q3}

The complete comparison with \omdt and \dtcontrol includes 21 models and 8 fixed depths, resulting in 168 different evaluations. As such, we present the resulting table as a CSV file available here: \url{https://anonymous.4open.science/r/dt-synthesis-cav-25-5B31/results.csv}.

\subsection{Runtimes}
\label{app:runtimes}

Figure~\ref{fig:appendix-omdt-runtimes} compares the runtimes between \dtpaynt and \omdt from experiments in Q1. Points above the diagonal indicate that \dtpaynt was faster. As can be seen, if the performance in terms of quality is similar, \dtpaynt still provides an advantage as it is able to produce the results significantly faster most of the time.

\setlength{\textfloatsep}{12pt}
\begin{figure}[t]
    \centering
    \resizebox{0.8\textwidth}{!}{\input{figures/plot-appendix-omdt-runtimes.pgf}}
    \vspace{-1em}
    \caption{Comparison in runtimes between \dtpaynt and \omdt. Only experiments where the normalized values of the trees found by the tools are within 10\% of each other.}
    \label{fig:appendix-omdt-runtimes}
\end{figure}

\subsection{Complete Q2 graph}

Figure~\ref{fig:appendix-complete-Q2} shows the same plot as Fig.~\ref{fig:dtcontrol-bar-plot}, but it includes all of the 21 models used in this experiment.

\setlength{\textfloatsep}{12pt}
\begin{figure}[t]
    \centering
    \resizebox{\textwidth}{!}{\input{figures/plot-appendix-q2-complete.pgf}}
    \vspace{-1em}
    \caption{Complete graph for Q2, including all models from our benchmark set.}
    \label{fig:appendix-complete-Q2}
\end{figure}

% \section{Big table}

% \begin{table}[h]
% \setlength{\tabcolsep}{1pt}

% \input{tables/big-omdt-table}

% \caption{Comparison with \omdt.}
% \label{tab:appendix-big-omdt-table}
% \end{table}

\end{document}